\documentclass[journal]{IEEEtran}

\usepackage{amssymb}
\usepackage{amsfonts}
\usepackage{amsmath}
\usepackage{amsthm}
\usepackage[all]{xy}
\usepackage{latexsym}
\usepackage{enumerate}
\usepackage{algorithm}
\usepackage{xcolor}

\usepackage{dsfont}
\usepackage{bbm}
\usepackage{fancyhdr}

\usepackage{graphicx}
\usepackage{authblk}

\DeclareMathOperator{\imag}{Im}
\DeclareMathOperator{\linspan}{Span}
\DeclareMathOperator{\supp}{Supp}
\DeclareMathOperator{\setv}{Set}

\newcommand{\bbN}{\mathbb{N}}
\newcommand{\bbNp}{\mathbb{N}^*}
\newcommand{\bbZ}{\mathbb{Z}}

\newcommand{\ef}{\mathbb{F}}
\newcommand{\efpi}{\ef_p}
\newcommand{\efq}{\ef_q}
\newcommand{\eftwo}{\ef_2}
\newcommand{\eftwomm}[1]{\ef_{2^{#1}}}
\newcommand{\eftwom}{\eftwomm{m}}

\newcommand{\bs}{\boldsymbol}

\newcommand{\veps}{\varepsilon}
\newcommand{\tl}{\tilde}

\newcommand{\RM}{\mathrm{RM}}
\newcommand{\bch}{\mathrm{BCH}}
\newcommand{\ebch}{\mathrm{eBCH}}

\newcommand{\bbch}[1]{\bch(#1)}
\newcommand{\bbchd}{\bbch{d}}
\newcommand{\bbchdone}{\bbch{d-1}}
\newcommand{\ebbch}[1]{\ebch(#1)}
\newcommand{\ebbchd}{\ebbch{d}}
\newcommand{\wh}{\mathrm{w}_{\mathrm{H}}}

\newcommand{\pdvar}[1]{P_{#1}}
\newcommand{\pd}{\pdvar{d}}
\newcommand{\pdone}{\pdvar{d-1}}
\newcommand{\polsrvar}[1]{\eftwom[X_1,\ldots,X_m]_{\leq #1}}
\newcommand{\polsr}{\polsrvar{r}} 

\newcommand{\bx}{\bs{x}}

\newcommand{\boldf}{\bs{f}}
\newcommand{\boldh}{\bs{h}}

\newcommand{\roots}{\mathrm{Roots}}
\newcommand{\one}{\bs{1}}

\newcommand{\grobner}{Gr\"obner}
\newcommand{\fld}[2]{{#1}_{#2}^{\mathrm{fld}}}
\newcommand{\vct}[2]{{#1}_{#2}^{\mathrm{vct}}}
\newcommand{\wt}{\mathrm{wt}}

\newcommand{\tr}{\mathrm{Tr}}
\newcommand{\trv}[2]{\tr^{#1}_{#2}}
\newcommand{\nrm}{\mathrm{N}}
\newcommand{\normv}[2]{\nrm^{#1}_{#2}}

\newtheorem{proposition}{Proposition}[section]
\newtheorem{definition}[proposition]{Definition}
\newtheorem{corollary}[proposition]{Corollary}
\newtheorem{lemma}[proposition]{Lemma}
\newtheorem{theorem}[proposition]{Theorem}
\newtheorem{remark}[proposition]{Remark}
\newtheorem{example}[proposition]{Example}

\title{Efficient Algorithms for Constructing Minimum-Weight Codewords
in Some Extended Binary BCH Codes}
\author{ 
\IEEEauthorblockN{Amit Berman, Yaron Shany, and Itzhak Tamo}
\thanks{Amit Berman and Yaron Shany are with Samsung  Semiconductor Israel R\&D Center, 146 Derech
Menachem Begin St., Tel Aviv 6492103, Israel. Emails: \{amit.berman,
yaron.shany\}@samsung.com}
\thanks{Itzhak Tamo is with the Department of Electrical Engineering-Systems, Tel Aviv
University, Tel Aviv 6997801, Israel, and with Samsung  Semiconductor Israel R\&D Center, 146 Derech
Menachem Begin St., Tel Aviv 6492103, Israel. Email: zactamo@gmail.com}
}

\begin{document}
\maketitle

\begin{abstract}
We present $O(m^3)$ algorithms for specifying the support of
minimum-weight codewords of extended binary BCH codes of length
$n=2^m$ and designed distance $d(m,s,i):=2^{m-1-s}-2^{m-1-i-s}$ for some
values of $m,i,s$, where $m$ may grow to infinity. Here, the support is
specified as the sum of two sets: a set of $2^{2i-1}-2^{i-1}$ elements,
and a subspace of dimension $m-2i-s$, specified by a basis.

In some detail, for designed distance $6\cdot 2^j$,
$j\in\{0,\ldots,m-4\}$, we have a deterministic algorithm for even
$m\geq 4$, and a probabilistic algorithm with success probability
$1-O(2^{-m})$ for odd $m>4$.  For designed distance $28\cdot 2^j$,
$j\in\{0,\ldots, m-6\}$, we have a probabilistic algorithm with success
probability $\geq\frac{1}{3}-O(2^{-m/2})$
for even $m\geq 6$. Finally, for designed 
distance $120\cdot 2^j$, $j\in\{0,\ldots,m-8\}$, we have a
deterministic algorithm for $m\geq 8$ divisible by $4$.  
We also show how Gold
functions can be used to find the support of minimum-weight words for
designed distance $d(m,s,i)$ (for $i\in\{0,\ldots,\lfloor
m/2\rfloor\}$, and $s\leq m-2i$) whenever $2i|m$.
  
Our construction builds on results of Kasami and Lin (IEEE Trans.~IT,
1972), who proved that for extended binary BCH codes of designed
distance $d(m,s,i)$ (for integers $m\geq 2$, $0\leq i\leq \lfloor
m/2\rfloor$, and $0\leq s\leq m-2i$), the minimum distance equals the
designed distance. The proof of Kasami and Lin makes use of a
non-constructive existence result of Berlekamp (Infom.~Contrl., 1970),
and a constructive ``down-conversion theorem'' that converts some
words in BCH codes to lower-weight words in BCH codes of lower
designed distance. Our main contribution is in replacing the
non-constructive counting argument of Berlekamp by a low-complexity
algorithm. 

In one aspect, the current paper extends the results of Grigorescu and
Kaufman (IEEE Trans.~IT, 2012), who presented explicit minimum-weight
codewords for extended binary BCH codes of designed distance \emph{exactly}
$6$ (and hence also for designed distance $6\cdot 2^j$, by a
well-known ``up-conversion theorem''), as we cover more cases of the
minimum distance. In fact, we prove that the codeword
constructed by Grigorescu and Kaufman is a special case of the
current construction. However, the minimum-weight codewords we
construct do not generate the code, and are not affine generators,
except, possibly, for a designed distance of $6$.
\end{abstract}

\section{Introduction}

\subsection{Motivation and known results}
Determining the true minimum distance of BCH codes is a notoriously
hard problem that has been studied extensively in the coding theory
literature; see, e.g., \cite{Be70}, \cite{KaLi72}, \cite{AS94},
\cite{Li17}, and the references therein.  While the minimum distance
can be sometimes settled by non-constructive arguments (e.g., as in
\cite{Be70}), there is a recent interest in exhibiting explicit
minimum-weight codewords.  

A {\it single-orbit affine generator} (or simply,
an {\it affine generator}) of a linear code (if exists) is
a codeword whose orbit under the action of the affine group
generates the code as a vector space.
Applications to property testing motivated papers such as
\cite{GKS12}, \cite{GK12}, and \cite{MoSo18} to search for a single
low-weight\footnote{Throughout, ``weight'' and
``distance'' refer to Hamming weight and Hamming distance,
respectively.} affine generator for binary BCH codes. Both
\cite[pp.~15--16]{GKS12} and \cite[Problem 2]{GK12} raised the
following two questions: 

\begin{description}

\item[Q1] What is an explicit low-weight codeword of a (primitive,
narrow-sense) binary BCH code of a given length and designed minimum
distance?

\item[Q2] (Harder) What is an explicit basis consisting of low-weight
codewords for the above BCH code?

\end{description}

There are some cases where the answer to Q1 is well-known. For
example, the binary vector of length $2^m$ supported on an affine
subspace $V\subseteq \eftwom$ (where for a prime power $q$, $\efq$ is
the finite field of $q$ elements, and $m\in \bbNp$) is a
minimum-weight codeword of the extended binary BCH code of designed
distance $|V|$, since it is in the Reed--Muller subcode (see, e.g.,
\cite[Ch.~13]{MacSl}).\footnote{Note that in general, such a codeword
cannot be an affine generator, as it is in proper affine-invariant
subcode.} Also, when $\eftwom^*$ has a non-trivial 
subgroup $G$ (i.e., when $2^m-1$ is not a prime), the binary
vector of length $2^m-1$ supported on $G$ is a minimum-weight
codeword of the binary BCH code of designed distance $|G|$ (this can
be easily verified directly; see also \cite[Coro.~9.4]{PeWe}).

A na\"ive probabilistic algorithm for finding minimum-weight words
when it is known that the designed distance equals the actual distance
$d$, is
to draw uniformly vectors of weight $t+1$, where $t:=\lfloor (d-1)/2
\rfloor$, and apply any efficient decoding algorithm for the BCH
code. However, it can be verified that the success probability of this
algorithm is about\footnote{This approximation is valid when
the number of minimum-weight words is close to the binomial
distribution. This is indeed the case if $t$ is fixed and the length 
grows \cite{KL99}.} $1/t!$, making it useless for moderately high
values of $t$. 

In \cite{Au96}, the problem of finding low-weight codewords in a
cyclic code was translated to the problem of solving a large system of
polynomial equations obtained from Newton's identities. It is then
suggested to use elimination theory through \grobner{}-bases and
Buchberger's algorithm for solving these equations. It is shown in
\cite{Au96} that this method is
useful for settling some questions regarding short codes (e.g.,
finding the exact minimum distance and the number of minimum-weight
codewords in a dual BCH code of length $63$). However, the
complexity is in general exponential in the code length. 

In \cite{AS94}, there is an algorithm for finding the locator
polynomial of {\it idempotent} minimum-weight codewords of binary BCH
codes. However, the results of the current paper cover 
weights that are not covered in \cite{AS94}.\footnote{For example,
considering words of weight $27$, not all even
values of $m$ are covered in \cite[Table I]{AS94}, while in the current
paper we present a low-complexity probabilistic algorithm for all even
$m$ (see Section \ref{sec:ci3}).} Indeed, the methods of the
current paper are entirely different, and are not limited to
idempotent minimum-weight words. 

An algorithm for finding minimum-weight words in an arbitrary
linear code appears in \cite{CC98}. While this algorithm performs well
in practice, and was used to find minimum-weight words in BCH
codes of length $511$, its worst-case complexity is exponential
in the code length.

A partial answer to Q2 appears in
\cite{GKS12}, where it is proved that under certain conditions, there
exists a low-weight affine generator for long binary BCH codes of a
fixed designed distance. A progress toward an explicit low-weight affine
generator was made in \cite{GK12}, 
where the support of an affine generator of weight $6$ is specified
explicitly for extended double error-correcting BCH codes.  The
results of \cite{GK12} were generalized in \cite{MoSo18} to a larger
class of extended cyclic codes\footnote{The zeros of
the codes in \cite{MoSo18} are $\alpha,\alpha^{2^i+1}$ and their
conjugates, for some primitive $\alpha$ in $\eftwom$, where
$\gcd(m,i)=1$. As commented in \cite{MoSo18}, in this sense, the codes
are related to the Gold functions $x\mapsto x^{2^i+1}$.} of minimum
distance $6$. 

In addition to the theoretical motivation raised in Q1, an additional 
motivation comes from the need for debugging 
coding schemes involving BCH codes for very low
probabilities of error. For this purpose, it is required
to generate artificial error vectors that assure that the decoder
encounters some (loosely speaking) rare event, as reaching these rare
events through a Monte Carlo simulation is infeasible. 

One example of such a rare event is a long sequence of zero {\it
discrepancies} in the Berlekamp-Massey algorithm \cite{Be66},
\cite{Ma69} for a non-zero error, that still results in a successful
decoding. Another 
example is in decoding generalized concatenated
codes\footnote{See, e.g., \cite{SSBZ08} and the references therein.}
with inner BCH codes. Here, such a rare event can be an error vector that
is undetected by the decoder of a first BCH code, but then successfully
decoded when the multi-stage decoding algorithm can finally decode it
using a second BCH decoder, for a second BCH code of about twice the
designed minimum distance. In both examples, explicit minimum-weight
codewords can be used. In fact, minimum-weight codewords constructed
by the methods of the current paper were used for debugging actual
implemented codes.

\subsection{Results and methods}
We present $O(m^3)$ algorithms for constructing the support of
minimum-weight codewords of extended binary BCH codes of length
$n=2^m$ and designed distance $d(m,s,i):=2^{m-1-s}-2^{m-1-i-s}$ for some
values of $m,i,s$ (to be specified shortly), where $m$ may grow to
infinity. Although $d(m,s,i)$
may be exponential in $m$, specifying the support amounts to
specifying two sets $X,B\subseteq \eftwom$ with $|X|+|B|=O(m)$, in
the following way.

For $S_1,S_2\subseteq \eftwom$, let $S_1+S_2:=\{x+y|x\in S_1, y\in
S_2\}$. The support of size $d(m,s,i)$ is given by
$X+\linspan_{\eftwo}(B)$, where $X,B\subseteq \eftwom$,  
$|X|=2^{2i-1}-2^{i-1}$, $|B|=m-2i-s$, and the elements of $B$ are
$\eftwo$-linearly independent. Note that we consider $i$ as constant,
while $m$ may grow to infinity.

We build on results of Kasami and Lin \cite{KaLi72} in a way that we
shall now describe. It was shown in \cite{KaLi72} that for integers
$m\geq 2$, $0\leq i\leq \lfloor m/2\rfloor$, and $0\leq s\leq m-2i$,
the true minimum distance of the extended BCH code of designed
distance $d(m,s,i)$ is indeed $d(m,s,i)$. The proof in \cite{KaLi72}
is based on two results: (1) a constructive ``down-conversion
theorem,'' \cite[Theorem 1]{KaLi72}
that enables to convert certain codewords in a BCH code to codewords
of lower weights in another BCH code of a lower designed distance, and
(2) a non-constructive existence result of Berlekamp \cite{Be70} for
the case $s=0$.

Our main contribution is in replacing Berlekamp's counting argument by
$O(m^3)$ algorithms for producing minimum-weight codewords in extended
BCH codes of designed distance $d(m,0,i)=2^{m-1}-2^{m-1-i}$ for
$i=2,3,4$, under some conditions on $m$ to be described ahead. The
output of the algorithm are the above two sets $X$ and $B$
(with $|X|=2^{2i-1}-2^{i-1}$, $|B|=m-2i$),
such that the support is $X+\linspan_{\eftwo}(B)$. 

We can then use the
down-conversion theorem with $X$ and $B$ to obtain, for all integer
$0\leq s\leq m-2i$, sets $X_s,B_s\subseteq \eftwom$, with
$|X_s|=2^{2i-1}-2^{i-1}$ and $|B_s|=m-2i-s$, such that
$X_s+\linspan_{\eftwo}(B_s)$ supports a minimum-weight codeword of the
extended BCH code of designed distance $d(m,s,i)$, where $X_0=X$
and $B_0=B$. For completeness, we also re-prove the down-conversion
theorem in a way that we find more intuitive: instead of working with
locator polynomials, we use the low-degree evaluated
polynomials in the definition of the BCH code as a subfield subcode of
a Reed--Solomon code.

More generally, we show that minimum-weight codewords of extended BCH
codes of designed distance $d(m,0,i)$ (for all $0\leq i\leq \lfloor
m/2\rfloor$) can be obtained in complexity $O(m^3)$ once a solution
with $\eftwo$-independent entries is found to a system of $i-1$
$\eftwo$-multilinear equations in $2i$ variables.
The equations are obtained using the following idea. 
In \cite{Be70}, it is shown that the intersection of the
second-order Reed--Muller code and the extended BCH code of designed
distance $d(m,0,i)$ contains a word of weight $d(m,0,i)$. For considering
Reed--Muller codes and extended BCH codes with the same coordinate
labeling, we must fix 
a basis for $\eftwom/\eftwo$. The idea is that instead of fixing a
basis and then trying to find the quadratic Boolean function in the
intersection, it is easier to fix the Boolean function and then
to search for a basis. In detail, one can fix the Boolean function
$(x_1,\ldots,x_m)\mapsto x_1x_2+x_3x_4+\cdots +x_{2i-1}x_{2i}$, and then
to search for a basis for which the evaluation vector of this function
becomes the evaluation vector of a polynomial from $\eftwom[X]$ of a
low enough degree.

For $i=2$ and even $m\geq 4$, we have a closed-form solution for the
equations with independent coordinates. For $i=2$ 
and odd $m>4$, we have a closed-form solution whose coordinates are
independent with probability $1-O(2^{-m})$. For $i=3$ and even $m\geq
6$, we have a probabilistic solution with success probability 
$\frac{1}{3}-O(2^{-m/2})$. Finally, for $i=4$ and $m\geq 8$ divisible
by $4$, we have a  closed-form solution.

We also show that for all $i$, minimum-weight codewords of weight
$d(m,s,i)$ can be obtained by using Gold functions (see, e.g.,
\cite{Le06}) if $2i|m$. We note that this result is different from that
\cite{MoSo18}, where the considered \emph{codes} are related to Gold
functions, have a minimum distance of $6$, and are in general not BCH
codes. 

Finally, we prove that the weight-$6$ codeword from \cite{GK12} is a
special case of the current construction,
corresponding to some solution of the system of equations for 
$i=2$. 
We comment that using an ``up-conversion theorem'' 
\cite[Theorem 9.5]{PeWe}, the codeword of \cite{GK12} can also be
used to reach all weights $6\cdot 2^j$ corresponding to
$d(m,s,2)$. Nevertheless, it is of interest to describe
solutions of the equations for $i=2$, for two main reasons: 1.~This is
the simplest example of the system of equations, and solutions for
this case provide insight for higher values of $i$, and
2.~such solutions may provide additional examples of
minimum-weight words, on top of those of \cite{GK12}. For
completeness, we also re-prove the up-conversion theorem in a more
direct  way, working with the evaluated polynomials instead of locator
polynomials.  

To summarize, our main contributions are as follows:

\begin{itemize}

\item We prove that minimum-weight codewords of extended BCH codes of
designed distance $d(m,s,i)$ can be obtained once a solution with
$\eftwo$-independent entries is found to a system of $i-1$
$\eftwo$-multilinear equations in $2i$ variables. The complexity of
moving from a solution of the equations to the actual support of a
minimum-weight codeword is $O(m^3)$.

\item We present solutions for $i=2,3,4$ as follows:

\begin{itemize}

\item For $i=2$, we have a closed-form solution with independent
coordinates for even $m\geq 4$. For odd $m>4$, we have a closed-form
random solution that solves the equations with probability $1$, and
has independent coordinates with probability $1-O(2^{-m})$. The
weights covered in this case are $6\cdot 2^j$,
$j\in\{0,\ldots,m-4\}$.

\item For $i=3$ and even $m\geq 6$, we have a probabilistic algorithm
with success probability at least $1/3-O(2^{-m/2})$ for finding a
solution with independent coordinates. The weights covered in this
case are $28\cdot 2^j$, $j\in\{0,\ldots, m-6\}$. 

\item For $i=4$ and $m\geq 8$ divisible by $4$, we have a closed-form
solution with independent coordinates. The weights covered in this
case are $120\cdot 2^j$, $j\in\{0,\ldots,m-8\}$.

\end{itemize}

\item When $2i|m$, we prove that Gold functions can be used to obtain
codewords of weight $d(m,s,i)$.

\item We re-prove the down-conversion theorem \cite[Theorem 1]{KaLi72}
and the up-conversion theorem \cite[Theorem 9.5]{PeWe} 
in a different way, using the evaluated polynomials
instead of locator polynomials. 

\end{itemize}

The different cases of $i$ and $m$ that are covered in the paper are
summarized in Table \ref{table:cover}. 

\begin{table*}[h]
\centering
{\small
\begin{tabular}{|c|c|c|c|c|}
\hline
$i$ & weights covered & $m$ &  where & remarks\\
\hline
$2$ & $6\cdot 2^j$,  & $\geq 4$, even & Prop.~\ref{prop:i2} & deterministic \\
\cline{3-5}
& $j\leq m-4$ & $\geq 5$, odd & Prop.~\ref{prop:i2odd} 
& probabilistic, success prob.~$1-O(2^{-m})$\\ 
\cline{3-5}
& & $m=ab$,  & 
Prop.~\ref{prop:constz} & deterministic\\
&& $\min\{a,b\}\geq 2$,&& \\
&& $\gcd(a,b)=1$&&\\
\hline
$3$ & $28\cdot 2^j$, & $\geq 6$, even & Prop.~\ref{prop:i3evenm} & probabilistic,
success prob.~$1/3-O(2^{-m/2})$\\ 
\cline{3-5}
& $j\leq m-6$ & $\geq 7$, odd & App.~\ref{app:heuristics} & heuristic\\
\hline
$4$ & $120\cdot 2^j$ & $\geq 8$, divisible by $4$ &
Prop.~\ref{prop:i4} & 
deterministic\\
& $j\leq m-8$ &&&\\
\hline
any & & divisible by $2i$ & Prop.~\ref{prop:gold} & deterministic,
via Gold functions\\
\hline
\end{tabular}
}
\caption{Cases covered in the paper.}
\label{table:cover}
\end{table*}

\subsection{Organization}
Section \ref{sec:preliminaries} includes some notation and definitions
used throughout the paper, and in particular some background on BCH and
Reed--Muller codes. As a preparation for the method used to construct
minimum-weight codewords, we demonstrate how this method can be used
to prove the well-known fact that Reed-Muller codes are subcodes of
extended BCH codes. 

In Section \ref{sec:conversion}, we recall the down-conversion and
the up-conversion theorems, and state them in a unified way that is more
suitable for the current paper. In Section
\ref{sec:main}, which is the heart of the paper, it is shown how the
problem for $s=0$ can be converted to a system of equations, and
solutions are provided in some cases. It is then shown how to use the
down-conversion theorem for moving from $s=0$ to arbitrary $s$, and some
concrete examples of the supports of minimum-weight codewords are
given.  A construction via Gold functions for the case $2i|m$ is
described in Section \ref{sec:gold}. Finally, Section
\ref{sec:conclusions} includes conclusions and open questions for
future research. 

The paper is supplemented by four appendices. In Appendix
\ref{app:conv}, we re-prove the conversion theorems in a new way. In
Appendix \ref{app:affine}, it is proved that for weight $>6$, the  
codewords constructed in this paper are \emph{not} affine generators,
as they are in a proper affine-invariant subcode.
In Appendix \ref{app:gk},
we show that the weight-$6$ affine generator of \cite{GK12} is a
special case of the current results, in the sense that it can be
obtained by an appropriate solution for the equations for the case
$i=2$. In particular, this shows that for weight $6$, there exists a
solution to the equations that, after applying the down-conversion
theorem, results in an affine generator. 
Finally, in Appendix \ref{app:heuristics}, we present a heuristic
probabilistic algorithm for the case $i=3$, $m$ odd.

\section{Preliminaries}\label{sec:preliminaries} 
In this section, we introduce some notation and recall some
well-known facts about BCH and Reed--Muller codes. Throughout,
unless otherwise noted, all vectors are row vectors, and the $i$-th
element of a vector $\bs{y}$ is denoted by $y_i$.  Also, $(\cdot)^t$
stands for matrix transposition. We assume familiarity with
basic properties of finite fields. The background material may be found,
e.g., in \cite[Sec.~V.5]{Lang}. The interested reader may find a more
comprehensive study of finite fields in \cite{LN97}.

\subsection{BCH and Reed--Muller codes}

\begin{definition}
For a field $K$ and an integer $r\in \bbNp$, let 
$$
K[X]_{r}:= \{f\in K[X]|\deg(f)\leq r \}.
$$
\end{definition}

From this point on, we fix $m\in \bbNp$ and let $n:=2^m-1$.  We also
fix a primitive element $\alpha\in \eftwom$, and an enumeration
$\{\beta_1,\ldots,\beta_n\}$ of $\eftwom^*$.

\begin{definition}{(BCH codes)}
{\rm
For $d\in \bbNp$, $d<n$, let 
$$
\pd:=\big\{f\in \eftwom[X]_{n-d}|\forall x\in \eftwom^*:
f(x)\in\eftwo\big\}.
$$
For odd $d$, the primitive binary BCH code, $\bbchd$, of length 
$n$ and designed distance $d$, is defined by 
$$
\bbchd:=\big\{(f(\beta_1),\ldots,f(\beta_n))|f\in \pd
\big\}. 
$$
}
\end{definition}
Note that $\pd$ is an $\eftwo$-vector space, and hence
$\bbchd\subseteq \eftwo^n$ is a linear code. 
Also, because a polynomial of degree $\leq n-d$ cannot have more than
$n-d$ roots, it is clear that the minimum Hamming distance of $\bbchd$
is at least $d$. By the same argument, the map $f\mapsto
(f(\beta_1),\ldots,f(\beta_n))$ is an injective map $\pd\to\eftwo^n$,
and so $\dim(\bbchd)=\dim(\pd)$.\footnote{This observation does not
immediately provide bounds on $\dim(\bbchd)$, but, since this
dimension is irrelevant for the current paper, we will not elaborate
on this subject.} From this point on, ``BCH code'' will mean
``primitive binary BCH code''.  

Any polynomial $f=a_0+a_1X+\cdots +a_{n-1}X^{n-1}\in \eftwom[X]$ of
degree at most $n-1$ satisfies 
\begin{multline*}
\sum_{i=1}^n f(\beta_i) = \sum_{i=0}^{n-1}f(\alpha^i) =
\sum_{i=0}^{n-1}\sum_{j=0}^{n-1}a_j\alpha^{ij} =\\
\sum_j a_j\sum_i(\alpha^j)^i = n a_0=nf(0)=f(0),
\end{multline*}
where the last equality follows since $n$ is odd. This means that for 
$f\in \pd$, $f(0)=\sum_{i=0}^{n-1}f(\beta_i)$ is a ``parity bit'' for
the vector $(f(\beta_i))_i$ (in particular, $f(0)$ is binary).

\begin{definition}{(Extended BCH codes)}
{\rm
For even $d\in \bbNp$, the extended BCH code, $\ebbchd$, of length
$n+1=2^m$ and designed distance $d$, is defined by
\begin{multline*}
\ebbchd:=\\
\big\{(f(0),f(\beta_1),\ldots,f(\beta_n))|f\in \pdone
\big\}. 
\end{multline*}
}
\end{definition}
Note that from the discussion before the definition, we may write 
\begin{multline*}
\ebbchd=\\
\Big\{\big(\sum_{j=1}^n c_j,c_1,\ldots,
c_n\big)\Big|(c_1,\ldots,c_n)\in \bbchdone \Big\},
\end{multline*}
and therefore $\bbchdone$ can be obtained from $\ebbchd$ by puncturing the
first coordinate. Note that since $d-1$ is
odd by assumption and the minimum distance of $\bbchdone$ is at least
$d-1$, the minimum distance of $\ebbchd$ is at least
$d$.\footnote{Another way to see that the minimum distance of $\ebbchd$
is at least $d$ is by noting that we evaluate polynomials of degree
$\leq (2^m-1)-(d-1)=2^m-d$ on $2^m$ points.} 

Since for any $a,b\in \eftwom$ and
$f(X)\in\pd$, also $f(aX+b)\in \pd$, the affine group, consisting of
all permutations of the form $x\mapsto ax+b$ (with $a,b\in\eftwom$,
$a\neq 0$) is contained in the automorphism group of $\ebbchd$. As
this group acts transitively\footnote{This means that for any two
coordinates there exists an element in the group whose action takes
one coordinate to the other.} (in fact, for this we only need the
translation group $\{x\mapsto x+b\}$), any word of weight $d$ (if
exists) in $\ebbchd$ may be used to construct words of weight $d$ and
$d-1$ in $\bbchdone$. As the affine group contains a cycle of
length $n$ that fixes $0$ (consider $x\mapsto ax$ for primitive $a$),
this also shows that $\bbchd$ is a cyclic code under some ordering of
the coordinates.

We now turn to Reed--Muller codes. Let $V$ be the vector space of
functions $\eftwo^m\to\eftwo$ (recall that any such function is in
fact a polynomial function in $m$ variables). It can be verified that
the kernel of the map $\eftwo[X_1,\ldots,X_m]\to V$ which assigns to a
polynomial the associated polynomial function is 
the ideal $(X_1^2+X_1,\ldots,X_m^2+X_m)$.

\begin{definition}{(Reed--Muller codes)}
{\rm
For $r\in\{1,\ldots,m\}$, let
\begin{multline*}
M_r:=
\{X_{i_1}X_{i_2}\cdots X_{i_r}|\\ i_1,\ldots,i_r\in
\{1,\ldots,m\}\text{ are distinct}\}, 
\end{multline*}
and
$$
M_{\leq r}:=\{1\}\cup M_1 \cup \cdots\cup M_r.
$$
Let $\polsr$ be the $\eftwo$-linear span of $M_{\leq r}$. Fixing an
enumeration $\{\bx_1,\ldots,\bx_{2^m}\}$ of $\eftwo^m$, the
$r$-th-order Reed--Muller code, $\RM(r,m)$, is defined by
\begin{multline*}
\RM(r,m):=\big\{\big(f(\bx_1),\ldots,f(\bx_{2^m})\big)\big|\\ f\in
\polsr \big\}. 
\end{multline*}
}
\end{definition}
Note that the functions $\eftwo^m\to \eftwo$ corresponding to the
monomials in $M_{\leq r}$ are linearly
independent, because it can be verified that no polynomial in $\polsr$
is in $(X_1^2+X_1,\ldots,X_m^2+X_m)$.
Hence,
$$
\dim(\RM(r,m))=1+m+\binom{m}{2}+\cdots+\binom{m}{r}.
$$

In general, if $X=\{x_1,\ldots,x_{\ell}\}$ and $Y$ are
sets (for $\ell\in \bbNp$), and $f\colon X\to Y$ is a function, we will write
$\bs{f}:=(f(x_1),\ldots,f(x_{\ell}))$ for the {\bf evaluation vector}
of $f$. It should be always clear from the context what it the domain
of $f$. Also, if $Z$ is yet another set, and
$X\overset{f}{\to}Y\overset{g}{\to}Z$ are functions, we will write $\bs{g\circ
f}:=\big(g(f(x_1)),\ldots, g(f(x_{\ell}))\big)$ for the evaluation
vector of $g\circ f$.

\subsection{The trace map and dual bases}
To recall the connection between Reed--Muller codes and BCH codes, it
will be convenient to work with the trace map.
\begin{definition}
{\rm
The trace map, $\tr\colon \eftwom\to \eftwo$, is
defined by  
\begin{eqnarray*}
\tr(x) & := & x+\sigma(x)+\sigma^2(x)+\cdots+{\sigma^{m-1}}(x),\\
       & = & x+x^2+x^4+\cdots+x^{2^{m-1}},
\end{eqnarray*}
where $\sigma\colon x\mapsto x^2$ is the Frobenius automorphism. 
}
\end{definition}
Since a non-zero polynomial of degree $2^{m-1}$ cannot have $2^m$
roots, the trace is not identically zero,
and hence the mapping $(x,y)\mapsto \tr(xy)$ is a non-degenerate
bilinear form. As 
usual in this case of finite dimension, the map $\varphi\colon
x\mapsto(y\mapsto\tr(xy))$ is an isomorphism of $\eftwo$-vector spaces
between $\eftwom$ and its dual space (see, e.g., \cite[Theorem
VI.5.2]{Lang}). Let $b_1,\ldots,b_m$ be a basis of 
$\eftwom$ over $\eftwo$. We refer to the elements $b'_1,\ldots,b'_m$
that map under $\varphi$ to the dual basis of $b_1,\ldots b_m$ as the
{\bf trace-dual basis} (or simply as the {\bf dual basis}) of
$b_1,\ldots,b_m$. Since an isomorphism takes 
a basis to a basis, $b'_1,\ldots,b'_m$ is also a basis for $\eftwom$
over $\eftwo$. Note that the dual basis is uniquely defined by the
constraints $\tr(b_ib'_j)=\one(i=j)$ for all $i,j$. 

\subsection{Reed--Muller codes are subcodes of extended BCH codes}
With the basis $B=\{b_1,\ldots,b_m\}$, a function $f\colon
\eftwom\to\eftwo$ induces a function $\vct{f}{B}\colon \eftwo^m\to
\eftwo$ defined by
$\vct{f}{B}(x_1,\ldots,x_m):=f(x_1b_1+\cdots+x_mb_m)$. Conversely, a
function $g\colon\eftwo^m\to \eftwo$ induces a function
$\fld{g}{B}\colon \eftwom\to \eftwo$ defined by
$\fld{g}{B}(x_1b_1+\cdots+x_mb_m):=g(x_1,\ldots,x_m)$, and 
the functions $f\mapsto\vct{f}{B} $, $g\mapsto \fld{g}{B}$ are
inverses of each other. Note that if $B':=\{b'_1,\ldots,b'_m\}$ is the
dual basis of $B$, then for $f\colon x\mapsto \tr(b'_j x)$, we have 
$\vct{f}{B}(x_1,\ldots,x_m)=x_j$. In what follows, we will call the
function $(x_1,\ldots,x_m)\mapsto x_j$ (for $j\in\{1,\ldots,m\}$) the
$j$-th {\bf projection}.

The space of functions $\eftwom\to\eftwom$ is isomorphic
to the quotient $\eftwom[X]/(X^{2^m}+X)$. Write
$T(X):=X+X^2+\cdots+X^{2^{m-1}}$ for the lowest degree polynomial
corresponding to the trace function. As we have just seen, the function
$\eftwo^m\to \eftwo$ corresponding to the polynomial $T(b'_jX)$ is
mapped under $\vct{(\cdot)}{B}$ to the function corresponding to the
multivariate polynomial $X_j$.

The following proposition states the well-known fact that
Reed--Muller codes are subcodes of BCH codes. Using the trace
function, the proof (which appears to be new) does not require
identifying the zeros of punctured Reed--Muller codes as cyclic
codes. We include the proof in the paper as a warm-up to the method
used in Section \ref{sec:main}. 

\begin{proposition}\label{prop:subcode}
$\RM(r,m)\subseteq \ebbch{2^{m-r}}$. In particular, the
minimum distance of $\RM(r,m)$ is at least $2^{m-r}$.
\end{proposition}

\begin{proof}
Fix a basis $B=\{b_1,\ldots,b_m\}$ of $\eftwom/\eftwo$, and let
$B'=\{b'_1,\ldots,b'_m\}$ be its dual basis. It is sufficient to
show that for any monomial in $M_{\leq r}$, there 
exists a polynomial $f\in \pdvar{2^{m-r}-1}$ that, after applying
$\vct{(\cdot)}{B}$, defines the same function $\eftwo^m\to
\eftwo$. So, for $s\leq r$, let $i_1,\ldots, 
i_s$ be distinct indices in $\{1,\ldots,m\}$, and consider the
monomial $\mu(X_1,\ldots,X_m):=X_{i_1}\cdots X_{i_s}\in
\eftwo[X_1,\ldots, X_m]$. As 
mentioned above, the polynomial $f(X):=T(b'_{i_1}X)\cdots
T(b'_{i_s}X)\in \eftwom[X]$ defines the same function as
$\mu(X_1,\ldots,X_m)$ (after applying $\vct{(\cdot)}{B}$).

Next, we would like to bound the degree of the
lowest-degree representative of the image of $f$ 
in $\eftwom[X]/(X^{2^m}+X)$. Observe that $f$ is a linear combination
of monomials of the form
\begin{equation}\label{eq:mon}
X^{2^{j_1}+2^{j_2}+\cdots+2^{j_s}}
\end{equation}
for some (not necessarily distinct) $j_1,\ldots,j_s\in
\{0,\ldots,m-1\}$. We claim that after identifying $X^{2^m}$ with $X$,
the highest possible degree, $d_{\max}$, is given by 
\begin{equation}\label{eq:maxdeg}
d_{\max}=2^{m-1}+2^{m-2}+\cdots + 2^{m-s}.
\end{equation}

Write\footnote{Note that the right-hand side of the following equation
is \emph{not} the binary expansion of the left-hand side.}
$\sum_{i=1}^s 2^{j_i}=\sum_{\ell=0}^{m-1}a_{\ell}2^{\ell} +a_m  
2^m$, where at first
$a_m=0$, and $a_{\ell}:=|\{i|j_i=\ell\}|$ (note that $\sum_{i=0}^m a_i
=s$). Consider a sequence of modifications of 
the $a_{\ell}$ in which we pick some $a_{\ell}\geq 2$ (if exists) for
$\ell\in\{0,\ldots,m-1\}$, and update $a_{\ell}\leftarrow a_{\ell}-2$ and
$a_{\ell+1}\leftarrow a_{\ell+1}+1$.\footnote{For example, if
$m=7$, $s=4$, $j_1=6$, $j_2=5$, $j_3=j_4=4$. Then, at first
$(a_0,\ldots,a_7)=(0,0,0,0,2,1,1,0)$, and the progression is
$(0,0,0,0,\bs{2},1,1,0)\to (0,0,0,0,0,\bs{2},1,0)\to
(0,0,0,0,0,0,\bs{2},0)\to (0,0,0,0,0,0,0,1)$, where the bold numbers are
the $a_{\ell}\geq 2$ picked as described above.} 
If there exists some $\ell\leq
m-1$ with $a_{\ell}\geq 2$, each such step strictly decreases
$\sum_{i=0}^m a_i$, while keeping all $a_{\ell}$
non-negative. Hence, the process must terminate after a finite number
of iterations with all $a_{\ell}\leq 1$ for all
$\ell\leq m-1$. Also, the modifications do not change the value of 
$\sum_{\ell=0}^m a_{\ell}2^{\ell}$.  

Hence, there exist $k_1,k_2\in \bbN$ with
$k_1+k_2\leq s$, and distinct
$j'_1,\ldots,j'_{k_2}\in \{0,\ldots,m-1\}$, such that
\begin{equation}\label{eq:rep}
\sum_{i=1}^s 2^{j_i} = k_1 2^m+\sum_{i=1}^{k_2} 2^{j'_i}  
\end{equation}
(and $k_1+k_2<s$ if the original $j_i$ are not distinct).

Re-writing the degree in (\ref{eq:mon})
as in the right-hand side of (\ref{eq:rep}) and identifying $X^{2^m}$
with $X$, we obtain a monomial of degree
\begin{multline*}
k_1 +\sum_{i=1}^{k_2} 2^{j'_i} \stackrel{(*)}{\leq} k_1 +
\sum_{i=1}^{k_2} 2^{m-i} \leq \\ \sum_{i=1}^{k_2+k_1} 2^{m-i}
\stackrel{(**)}{\leq}  \sum_{i=1}^{s} 2^{m-i},
\end{multline*}
where $(*)$ holds since (\ref{eq:maxdeg}) is clearly
valid for distinct $j_k$'s, while $(**)$ follows from $k_1+k_2\leq
s$.  This proves (\ref{eq:maxdeg}).
Now the proof follows immediately from (\ref{eq:maxdeg}), since
$$
\sum_{i=1}^{s} 2^{m-i}\leq \sum_{i=1}^{r} 2^{m-i} =
2^{m}-2^{m-r}=n-(2^{m-r}-1).   
$$
\end{proof}

Note that since the evaluation vector of the monomial $X_1\cdots X_r$
clearly has a Hamming weight of $2^{m-r}$, it follows from Proposition
\ref{prop:subcode} that the minimum distance of $\RM(r,m)$ is exactly
$2^{m-r}$. 

\section{The conversion theorems}\label{sec:conversion} 
In this section, we state the down- and up- conversion theorems
(\cite[Theorem 1]{KaLi72} and \cite[Theorem 9.5]{PeWe}, resp.) in a
unified way that is somewhat different from that of \cite{KaLi72},
\cite{PeWe}, and will be useful for the current paper. Rather than
showing that the current versions indeed follow from \cite[Theorem
1]{KaLi72} and \cite[Theorem 9.5]{PeWe}, it is simpler to prove them
directly, using a different approach, and for completeness the new
proofs appear in Appendix \ref{app:conv}.

Recall that if $q$ is a prime power and $L$
is a finite extension of $\efq$, then an 
$\efq$-{\bf linearized polynomial} over $L$ is a polynomial of the form
$\sum_{i=0}^{\ell} a_i X^{q^i}\in L[X]$ for some $\ell\in
\mathbb{N}$. Recall also that if $K\subseteq L$ are finite fields, and
$V\subseteq L$ is a $K$-vector space, than the annihilator polynomial of $V$,
$A(X):=\prod_{\gamma\in V}(X-\gamma)$, is a $K$-linearized
polynomial.

We begin with a modified version of \cite[Theorem 1]{KaLi72}.
We refer to this theorem as the down-conversion theorem, since it
states that under certain conditions, words of weight $w$ in
$\ebbch{d}$ can be converted to words of weight $w/2^s$ in
$\ebbch{\lceil d/2^s\rceil}$ (see the theorem for a precise
statement). From this point on, $\wh(\cdot)$ will stand for the
Hamming weight. Also, for a function $f$ and a subset $S$ of the
domain of $f$, we let $f(S):=\{f(s)|s\in S\}$. Finally, for two sets
$S_1,S_2$, we write $S_1\smallsetminus S_2:=\{x\in S_1|x\notin S_2\}$
for the set difference.

\begin{theorem}{(The down-conversion theorem, \cite[Theorem
1]{KaLi72})}\label{thm:dconversion} 
For an even integer $2\leq d\leq 2^m-2$ with $\lceil d/2^s 
\rceil$ even, suppose that a polynomial $f\in \pdone$ takes constant 
values on cosets of an $\eftwo$-subspace $V\subseteq \eftwom$ of
dimension $s$. Then there exists a polynomial $g(X)\in\eftwom[X]$ such
that $f(X)=g(A(X))$, where $A(X)$ is the annihilator 
of $V$, and there exists a constant $c\in \eftwom^*$ such
that for 
$$
h(X):=c\cdot g(X)\cdot \prod_{\gamma\in\eftwom \smallsetminus
A(\eftwom)}(X-\gamma),
$$
the evaluation vector
$\boldh:=(h(\beta))_{\beta\in \eftwom}$ is a codeword of
$\ebbch{\lceil d/2^s \rceil}$. Moreover, writing
$\boldf:=(f(\beta))_{\beta\in\eftwom}$, 
we have $\wh(\boldh)=\wh\big(\boldf\big)/2^s$. Finally, if $S\subseteq
\eftwom$ is the support of $\boldf$, then the support of $\boldh$ is
$A(S)$. 
\end{theorem}

For the proof, see Appendix \ref{app:proofdc}. Note that in the
theorem, given the codeword $\boldf$, one can construct the codeword $\boldh$
simply by calculating $A(S)$, without the need to find $g(X)$.

We have the following immediate corollary to Theorem \ref{thm:dconversion}.

\begin{corollary}\label{coro:dconversion}
Maintaining the notations of Theorem \ref{thm:dconversion}, if
$\wh(\boldf)=d$, so that $\boldf$ is a
minimum-weight codeword of $\ebbch{d}$, then $2^s|d$, and
$\wh(\boldh)=d/2^s$, so that $\boldh$ is a minimum-weight codeword of
$\ebbch{d/2^s}$. 
\end{corollary} 

It is also interesting to work in the opposite direction: given a
minimum-weight codeword of $\ebbch{d}$ for some 
(loosely speaking) small $d$, one would like to obtain
minimum-weight codewords in $\ebbch{d\cdot 2^j}$ for some values of
$j$. The following ``up-conversion theorem'' appears in a different
form in \cite[Theorem 9.5]{PeWe}. 
For the theorem, recall (e.g., from
\cite{BK09}) that for an extension $K\subset L$ of finite fields with
$[L:K]=m$, and for a $K$-subspace $V\subseteq L$, the
{\bf image polynomial} $B(X)\in L[X]$ of $V$ is the unique monic
linearized polynomial of degree $|K|^{m-\dim(V)}$ such that
$B(L)=V$. 

\begin{theorem}{(The up-conversion theorem, \cite[Theorem
9.5]{PeWe})}\label{thm:uconversion}
For even $d$ and for $f\in\pdone$, suppose that the support of the
evaluation vector $\boldf:=(f(x))_{x\in\eftwom}\in 
\ebbch{d}$ is contained in
an $\eftwo$-subspace $U$. Let $g(X)\in \eftwom[X]$ be the interpolation
polynomial for the function $U\to \eftwom$, $x\mapsto
f(x)$, and let $B(X)$ be the image polynomial of $U$. Then, 
$\bs{g\circ B}\in \ebbch{d\cdot\frac{2^m}{|U|}}$, with weight
$\wh(\boldf) \cdot \frac{2^m}{|U|}$. In particular, if
$\wh(\boldf)=d$, then $\bs{g\circ B}$ is a minimum-weight codeword
of $\ebbch{d\cdot\frac{2^m}{|U|}}$. Also, if $S$ is the support of
$\boldf$, then the support of $\bs{g\circ B}$ 
is $B^{-1}(S)$.  
\end{theorem}

For the proof, see Appendix \ref{app:proofuc}.

\begin{remark}
{\rm
\begin{enumerate}

\item The up-conversion theorem is non-trivial when $\dim(U)$ is strictly
smaller than $m$. In particular, smaller values of $\dim(U)$ enable to
construct more values of the final weight.

\item Note that the down-conversion and up-conversion theorems define
inverse functions in the following sense. If we start with the support
$S$ of the codeword $\bs{f}$ of Theorem \ref{thm:dconversion}, which
is the union of cosets of the subspace $V$, we obtain by
the theorem a codeword with support $A(S)$, contained in the subspace
$U:=A(\eftwom)$, with image polynomial $B(X)=A(X)$. By Theorem
\ref{thm:uconversion}, we then obtain a codeword with support
$B^{-1}(A(S))=A^{-1}(A(S))=S$, where the last equality follows since
$S$ is a union of cosets of $V=\ker(A)$. A similar argument shows that
first ``going up'' and then ``going down'' does not change the
support.
\end{enumerate}
}
\end{remark}

\section{Minimum-weight words}\label{sec:main}
In this section, we present a method to construct minimum-weight words
that is based on converting the problem to the question of solving
a certain system of equations. First, the system of equations is derived in
Section \ref{sec:syst}. Then, explicit and probabilistic solutions are
given for certain parameters in Section \ref{sec:sol}. Finally, in
Section \ref{sec:usingdct}, we show how the down-conversion theorem
can be used to write the support explicitly, given the above
solutions. 

\subsection{The system of equations: $\ebbch{2^{m-1}-2^{m-1-i}}\cap
\RM(2,m)$}\label{sec:syst}
The down-conversion theorem (Theorem \ref{thm:dconversion}) states 
explicitly how, given the support of a codeword $\boldf\in\ebbchd$
that satisfies the conditions of the theorem, one can obtain the
support of a codeword $\boldh\in\ebbch{\lceil d/2^s \rceil}$ of weight
$\wh(\boldf)/2^s$. 

To prove the existence of words of weight $2^{m-s-1}-2^{m-s-1-i}$ in
$\ebbch{2^{m-s-1}-2^{m-s-1-i}}$ for certain values of $i$ and $s$, 
Kasami and Lin \cite{KaLi72} used the 
down-conversion theorem, together with a non-constructive existence result
of Berlekamp \cite{Be70}, which states that there exists a word of
minimum weight $2^{m-1}-2^{m-1-i}$ in $\ebbch{2^{m-1}-2^{m-1-i}}\cap
\RM(2,m)$, and this word satisfies the 
conditions of the down-conversion theorem.   

We continue in a similar fashion, replacing Berlekamp's
existence result by an algorithm
for specifying the support of a codeword of weight
$2^{m-1}-2^{m-1-i}$ in
$\ebbch{2^{m-1}-2^{m-1-i}}\cap \RM(2,m)$. Toward this end, we will show
that the problem is equivalent to finding a point $\bs{b}=(b_1, \ldots,
b_{2i})\in\eftwom^{2i}$ in the solution set of certain multivariate
polynomial equations.  

Recall from Section \ref{sec:preliminaries} that for a fixed basis
$B=\{b_1,\ldots, b_m\}$ of $\eftwom/\eftwo$, we have a bijection  
$$
\{\text{functions }\eftwom\to
\eftwo\}\mathrel{\mathop{\rightleftarrows}^{\vct{(\cdot)}{B}}_{\fld{(\cdot)}{B}}}
\{\text{functions }\eftwo^m\to \eftwo\} 
$$
and that if $\{b'_1,\ldots,b'_m\}$ is the dual basis of $B$, then
$\vct{\big(x\mapsto \tr(b'_jx)\big)}{B}$ is the $j$-th projection,
$(x_1,\ldots,x_m)\mapsto x_j$.  

By a simple induction \cite[Lemma 15.2.6]{MacSl},
\cite[p.~441]{MacSl}, the evaluation vector of the polynomial 
$X_1X_2+X_3X_4+\cdots+X_{2i-1}X_{2i}$ (for $i\in\{1,\ldots, \lfloor
m/2\rfloor\}$) has weight $2^{m-1}-2^{m-1-i}$.
This fact will be used in the following theorem.\footnote{For
completeness, here is a proof that does not use induction. Clearly the
polynomial $X_1X_2+X_3X_4+\cdots+X_{2i-1}X_{2i}$ evaluates to $1$ if
and only if an odd number of the monomials $X_1X_2,
\ldots,X_{2i-1}X_{2i}$ evaluate to $1$ and the remaining variables
$X_j$ for $j=2i+1,\ldots, m$ can assume any possible
values. Therefore, the weight equals  
\begin{multline*}
\sum_{\substack{j=1\\j\text{odd}}}^{i}\binom{i}{j}3^{i-j} \cdot
2^{m-2i}=
\\2^{m-2i} \frac{\sum_{j=0}^i\binom{i}{j} 3^{i-j} -
\sum_{j=0}^i\binom{i}{j}
(-1)^j3^{i-j}}{2}=
\\2^{m-2i-1}(4^i-2^i)=2^{m-1}-2^{m-1-i}. 
\end{multline*}
}

\begin{theorem}\label{thm:basis}
Let $B=\{b_1,\ldots,b_m\}$  and $B'=\{b'_1,\ldots,b'_m\}$ be a basis
of $\eftwom/\eftwo$ and its dual basis, respectively. Let $i\in
\{2,\ldots,\lfloor m/2\rfloor\}$, and let $f\in\eftwom[X]$
be a polynomial such that $\vct{f}{B'}$ is $(x_1,\ldots,x_m)\mapsto
x_1x_2+x_3x_4+\cdots+x_{2i-1}x_{2i}$. Then the evaluation vector
$\bs{f}$ is in $\ebbch{2^{m-1}-2^{m-1-i}}$ iff the following $i-1$
equations hold:
\begin{equation}\label{eq:main}
\sum_{\substack{j=1\\j\text{
odd}}}^{2i-1}\Big(b_j^{2^{\ell}}b_{j+1}+b_jb_{j+1}^{2^{\ell}}\Big)=0,\quad
\ell=1,2,\ldots,i-1. 
\end{equation}
By construction, $\bs{f}$ takes constant values on cosets of the space $V$
spanned by $b'_{2i+1},\ldots,b'_m$.   
\end{theorem}

Note that the equations (\ref{eq:main}) involve only $b_1,\ldots,b_{2i}$.

\begin{proof}
By definition, $f\equiv f_0(X):=T(b_1 X)T(b_2 X)+\cdots+T(b_{2i-1}
X)T(b_{2i} X)\mod (X^{2^m}+X)$. Similarly to the proof of Proposition
\ref{prop:subcode},
$\bs{f}\in\ebbch{2^{m-1}-2^{m-1-i}}$ iff $f_0 \bmod (X^{2^m}+X)$ 
has a degree of at most
$$
\partial_0:=2^{m}-1-(2^{m-1}-2^{m-1-i}-1)=2^{m-1}+2^{m-1-i}. 
$$

Consider a single summand $u_j:=T(b_jX)T(b_{j+1}X)$ in the definition of
$f_0$. After reduction modulo $X^{2^m}+X$, we have
\begin{eqnarray}
u_j \equiv & 
 \big(b_j^{2^{m-1}}b_{j+1}^{2^{m-2}} + b_j^{2^{m-2}}b_{j+1}^{2^{m-1}}\big)
X^{2^{m-1}+2^{m-2}} + \nonumber\\
 &\big(b_j^{2^{m-1}}b_{j+1}^{2^{m-3}} + b_j^{2^{m-3}}b_{j+1}^{2^{m-1}}\big)
X^{2^{m-1}+2^{m-3}} + \nonumber\\
 & \qquad \qquad \cdots   \nonumber\\
 &\big(b_j^{2^{m-1}}b_{j+1}^{2^{m-i}} + b_j^{2^{m-i}}b_{j+1}^{2^{m-1}}\big)
X^{2^{m-1}+2^{m-i}} + \nonumber\\
 & \text{terms of degree }\leq 2^{m-1}+2^{m-i-1}=\partial_0. \label{eq:tbj}
\end{eqnarray}

Clearly, the degree of $f_0(X)\bmod (X^{2^m}+X)$ is at most
$\partial_0$ iff the coefficient of $X^k$ is zero for all $k>\partial_0$,
that is (using (\ref{eq:tbj})), iff
\begin{eqnarray*}
\sum_{\substack{j=1\\j\text{
odd}}}^{2i-1}\big(b_j^{2^{m-1}}b_{j+1}^{2^{m-2}} +
b_j^{2^{m-2}}b_{j+1}^{2^{m-1}}\big) & = & 0\\
\sum_{\substack{j=1\\j\text{
odd}}}^{2i-1}\big(b_j^{2^{m-1}}b_{j+1}^{2^{m-3}} +
b_j^{2^{m-3}}b_{j+1}^{2^{m-1}}\big) & = & 0\\
 & \vdots & \\
\sum_{\substack{j=1\\j\text{
odd}}}^{2i-1}\big(b_j^{2^{m-1}}b_{j+1}^{2^{m-i}} +
b_j^{2^{m-i}}b_{j+1}^{2^{m-1}}\big) & = & 0.
\end{eqnarray*}

Recall that $\sigma\colon x\mapsto x^2$ is the Frobenius
automorphism, and so $\sigma^{-1}$ is just the square root
function. Applying $\sigma^{-(m-2)}$ to the first equation,
$\sigma^{-(m-3)}$ to the second equation, etc., we get
(\ref{eq:main}). 
\end{proof}

\begin{remark}\label{rem:alls}
{\rm
\begin{enumerate}

\item If a solution with $\eftwo$-independent coordinates is found for the 
system (\ref{eq:main}), then it can be arbitrarily completed to a
basis $B$, and this determines the dual basis $B'$. Hence, the main
question is to find a solution with independent coordinates.

\item The evaluation vector $\bs{f}$ in Theorem
\ref{thm:basis} takes constant values on cosets of the space $V$
defined in the theorem, and hence also on cosets of any subspace of
$V$. Since $\dim(V)=m-2i$, the down-conversion theorem (Theorem
\ref{thm:dconversion}) can be used to 
replace $\bs{f}$ by a word of weight $2^{m-1-s}-2^{m-1-i-s}$ for all
integer $s\leq m-2i$.\footnote{Note that $m-2i\leq m-1-i$ for integer
$i>0$, and so for $s$ in the 
above range, the exponent $m-1-i-s$ is assured to be non-negative.} We
discuss this in more detail ahead in Section \ref{sec:usingdct}.

\end{enumerate}
}
\end{remark}

In the following subsection, we will present solutions with
independent coordinates for (\ref{eq:main}) for some values of $m$ and
$i$ (where $m$ is allowed to grow to infinity), and these solutions
do not rely on \cite{Be70}. However, the question arises whether such
solutions exist for the values of $m,i$ that are
\emph{not} covered ahead. Theorem \ref{thm:converse} below, together
with the existence result of Berlekamp 
\cite{Be70}, assure that such a solution in fact exists for all
$m\geq 2$ and $0\leq i\leq \lfloor m/2\rfloor$. 

Before stating Theorem \ref{thm:converse}, we recall the relevant part
of a theorem of Dickson \cite[Theorem 15.2.4]{MacSl}.

\begin{theorem}{(Dickson's Theorem, \cite[Theorem
15.2.4]{MacSl})}\label{thm:dickson} 
Let $g\colon\eftwom\to\eftwo$ be defined by
$g(\bs{x})=Q(\bs{x})+L(\bs{x})+\veps$, where $Q\colon\eftwo^m \to
\eftwo$ is a quadratic form, 
$L\colon\eftwo^m \to \eftwo$ is a linear functional, and $\veps\in
\eftwo$. Then there exists an invertible matrix $R\in\eftwo^{m\times
m}$ such that for $h\colon \eftwo^m\to\eftwo$ defined by
$h(\bs{y}):=g(\bs{y}R)$ we have 
$$
h(\bs{y})=\sum_{j=1}^{\ell}y_{2j-1}y_{2j}+L_1(\bs{y})+\veps
$$
for some $\ell\in \bbNp$ and some linear functional
$L_1\colon\eftwom\to\eftwo$.  
\end{theorem}

The following theorem considers a certain basis $D$ for
$\eftwom/\eftwo$, and concludes the existence of another basis, $E$,
that may take the role of the dual basis $B'$ of Theorem
\ref{thm:basis}.
\begin{theorem}\label{thm:converse}
If for $g\in \eftwo[X_1,\ldots,X_m]_{\leq 2}$ there exists a basis
$D:=\{d_1,\ldots,d_m\}$ of $\eftwom/\eftwo$ such that the evaluation vector of
$\fld{g}{D}\colon \eftwom\to \eftwo$ is in $\ebbch{2^{m-1}-2^{m-1-i}}$
and has weight $2^{m-1}-2^{m-1-i}$, then there
exists another basis $E:=\{e_1,\ldots,e_m\}$ such that for
$h\colon\eftwo^m\to \eftwo$ defined by
$h(x_1,\ldots,x_m)=\colon x_1x_2+x_3x_4+\cdots+x_{2i-1}x_{2i}$, the
evaluation vector of $\fld{h}{E}$ is in $\ebbch{2^{m-1}-2^{m-1-i}}$, and
the first $2i$ elements in the dual basis $E'$ of $E$ satisfy
(\ref{eq:main}).  
\end{theorem}

\begin{proof}
Take a quadratic form $Q\colon\eftwo^m \to
\eftwo$, a functional $L\colon\eftwo^m\to \eftwo$ and $\veps\in
\eftwo$ such that for $g\colon \bs{x}\mapsto
Q(\bs{x})+L(\bs{x})+\veps$ and for some basis $D=\{d_1,\ldots,d_m\}$ of
$\eftwom/\eftwo$, the evaluation vector of $\fld{g}{D}$ is in
$\ebbch{2^{m-1}-2^{m-1-i}}$. By Theorem \ref{thm:dickson}, there
exists an invertible matrix $R\in\eftwo^{m\times 
m}$ such that for $h\colon \eftwo^m\to\eftwo$ defined by $h(\bs{y}):=g(\bs{y}R)$ we have 
$$
h(\bs{y})=\sum_{j=1}^{\ell}y_{2j-1}y_{2j}+L_1(\bs{y})+\veps
$$
for some $\ell\in \bbNp$ and some linear functional
$L_1\colon\eftwom\to\eftwo$. 

Hence, for the basis (written as a vector)
$\bs{e}=(e_1,\ldots,e_m):=(d_1,\ldots,d_m)R^t$ of $\eftwom/\eftwo$,
we get on one hand, for $E:=\{e_1,\ldots,e_m\}$,
$$
\fld{h}{E}(y_1e_1+\cdots+y_me_m)=\sum_{j=1}^{\ell}y_{2j-1}y_{2j}+L_1(\bs{y})+\veps,
$$
(by the definition of $\fld{h}{E}$), and, on the other hand, we claim
that $\fld{h}{E}(x)=\fld{g}{D}(x)$ for all $x\in 
\eftwom$. To see this, write $\bs{d}:=(d_1,\ldots,d_m)$, and
$x=\bs{e}\bs{x}^t=\bs{d}R^t\bs{x}^t=\bs{d}\big(\bs{x}R\big)^t$,
to get $\fld{g}{D}(x)=g(\bs{x}R)=h(\bs{x})=\fld{h}{E}(x)$.
So, because the evaluation vector of $\fld{g}{D}$ is in
$\ebbch{2^{m-1}-2^{m-1-i}}$, so is the evaluation vector of
$\fld{h}{E}$ (this is the same vector). Also, if the weight of this
vector is $2^{m-1}-2^{m-1-i}$, it follows from \cite[Theorem
15.2.5]{MacSl} that $\ell=i$.  

It is easily verified that the evaluation vector of
$\fld{(L_1+\veps)}{E}$ is in $\ebbch{2^{m-1}-2^{m-i-1}}$, and, since
the evaluation vector of $\fld{h}{E}$ is also in
$\ebbch{2^{m-1}-2^{m-i-1}}$, it follows that so is the evaluation
vector of $\fld{(\bs{y}\mapsto
\sum_{j=1}^{i}y_{2j-1}y_{2j})}{E}$. It now follows from Theorem
\ref{thm:basis} that the first $2i$ elements of the dual basis of $E$
must satisfy (\ref{eq:main}).  
\end{proof}

\subsection{Solutions to the system of equations (\ref{eq:main})}\label{sec:sol}
For $j\in \bbNp$, let 
$$
f_j(X_1,X_2) := X_1^{2^j}X_2+X_1X_2^{2^j}\in \eftwo[X_1,X_2]. 
$$ 
Note that $f_j$ is a homogeneous polynomial of degree $2^j+1$.

\subsubsection{The case $i=2$}\label{sec:ci2}
For $i=2$, we are looking for a point $\bs{b}=(b_1,b_2,b_3,b_4)$ with
$f_1(b_1,b_2)=f_1(b_3,b_4)$, 
that is, with  
\begin{equation}\label{eq:i2}
b_1^2b_2+b_1b_2^2 = b_3^2b_4+b_3b_4^2, 
\end{equation}
such that $b_1,b_2,b_3,b_4$ are $\eftwo$-linearly independent.

\hfill\\ \noindent {\bf $\bs{m}$ even}\hfill \\
In this case, $\ef_4\subseteq \eftwom$. 
\begin{proposition}\label{prop:i2}
Suppose that $m\geq 4$ is even. Let $c$ be a primitive element of
$\ef_4^*$. If $x,y\in \eftwom$ are linearly independent over $\ef_4$,
then $\bs{b}:=(x,y,cx,cy)$ has $\eftwo$-independent 
coordinates. Moreover, for all odd $j$ and all $x,y$, $\bs{b}$
satisfies  
\begin{equation}\label{eq:oddj}
f_j(b_1,b_2)=f_j(b_3,b_4).
\end{equation} 
In particular, (\ref{eq:i2}) holds.
\end{proposition}

\begin{proof}
If $a_1x+a_2y+a_3cx+a_4cy=0$ for some
$a_1,\ldots,a_4\in \eftwo$, then $x(a_1+a_3c)+y(a_2+a_4c)=0$, and the
linear independence of $x,y$ over $\ef_4$, as well as the linear
independence of $1,c$ over $\eftwo$, imply that $a_1=\cdots=a_4=0$.

Let $j\in \bbNp$ be odd. Note that for odd $j$, $2^j+1\equiv 0\mod (3)$,
and therefore 
$$
f_j(cx,cy)=c^{2^j+1}f_j(x,y)=f_j(x,y),
$$
so that $\bs{b}$ satisfies (\ref{eq:oddj}).
\end{proof}

Note that in Proposition \ref{prop:i2}, it is possible to choose,
e.g., $(x,y)=(1,\alpha)$, where $\alpha\in
\eftwom^*$ is primitive.\hfill\\

\noindent {\bf $\bs{m}$ odd}\hfill \\
When $m$ is odd, $2^m-1\not\equiv 0 \mod(3)$, so that
$\gcd(3,2^m-1)=1$, and $3$ has a 
multiplicative inverse in $\bbZ/(2^m-1)\bbZ$. Hence, a unique cubic
root exists for all $z\in \eftwom$.

\begin{proposition}\label{prop:i2odd}
Suppose that $m\geq 5$ is odd. Fix distinct $v,v'\in\eftwom^*$ 
and let 
\begin{multline*}
\bs{b} = (b_1,b_2,b_3,b_4) =\\ \Big(\sqrt[3]{\frac{v^2}{c+v}},
\sqrt[3]{\frac{(c+v)^2}{v}}, \sqrt[3]{\frac{v'^2}{c+v'}},
\sqrt[3]{\frac{(c+v')^2}{v'}}\Big), 
\end{multline*}
where $c$ is picked uniformly at random from $ \eftwom\smallsetminus\{v,v',0\}$. Then, 
$\bs{b}$ is a solution for (\ref{eq:i2}), and with probability at least
$1-O(2^{-m})$ its entries are  $\eftwo$-linearly independent. 
\end{proposition}

\begin{proof}
It can be verified that $b_1^2b_2=v$, $b_1b_2^2=c+v$, and similarly
for $b_3,b_4$ with $v'$ instead of $v$. Hence, both sides of
(\ref{eq:i2}) equal $c$. We proceed to show the claim about the linear
independence.  

Assume that there exist $a_i\in\eftwo$ not all zero, such that
$\sum_ia_ib_i=0$. Equivalently, 
\begin{eqnarray*}
a_1\sqrt[3]{\frac{v^2}{c+v}}+a_2\sqrt[3]{\frac{(c+v)^2}{v}}
=\\
a_3\sqrt[3]{\frac{v'^2}{c+v'}} + a_4\sqrt[3]{\frac{(c+v')^2}{v'}}.
\end{eqnarray*}
Dividing by   $\sqrt[3]{\frac{v^2}{c+v}\frac{v'^2}{c+v'}}$  we obtain
\begin{eqnarray*}
 \Big(a_1+a_2\frac{c+v}{v}\Big)\sqrt[3]{\frac{c+v'}{v'^2}} = 
\Big(a_3+a_4\frac{c+v'}{v'}\Big)\sqrt[3]{\frac{c+v}{v^2}}. 
\end{eqnarray*}
Raising to the power of $3$ yields 
\begin{eqnarray*}
 \Big(a_1+a_2\frac{c+v}{v}\Big)^3\frac{c+v'}{v'^2} & = &
\Big(a_3+a_4\frac{c+v'}{v'}\Big)^3\frac{c+v}{v^2}. 
\end{eqnarray*}
This is a polynomial of degree at most $4$ in the variable $c$, where
the coefficient of $c^4$ and the constant term are
$$
\frac{a_2}{v^3v'^2} +
\frac{a_4}{v^2v'^3},  
\frac{a_1+a_2}{v'}+\frac{a_3+a_4}{v},
$$ 
respectively (note that, e.g.,
$(a_1+a_2)^3=(a_1+a_2)$, as the $a_i$ are
binary). They are both zero together iff $a_i=0$ for all $i$, as
$v\neq v'$, and $v,v'\neq 0$.  
We conclude that in either of the $2^4-1$ choices of the $a_i$'s,
the resulting polynomial is nonzero of degree at most $4$, and
therefore there are at most $4$ possible ``bad'' choices for $c$ for
which the $b_i$'s are $\eftwo$-linearly dependent. We note that this
analysis gives a positive probability for odd $m\geq 7$. For $m=5$, a
more detailed account of the degree of the polynomial in $c$ for the
various choices of $a_1,\ldots a_4$ can be used to show a positive
probability. We omit the details.
\end{proof}

\hfill\\ \noindent {\bf Composite $\bs{m}$}\hfill \\
So far, for $i=2$ we have a deterministic solution for even $m$, and a
probabilistic solution with success probability $1-O(2^{-m})$ for odd
$m$. The following is an additional deterministic construction, which
covers some values of composite $m$, both even and odd. For the odd
values of $m$ covered by Proposition \ref{prop:constz} below, this
provides a deterministic algorithm, instead of the previous
probabilistic algorithm.

\begin{proposition}\label{prop:constz}
Let $\ell,t\geq 2$ be coprime integers. Let $m:=\ell t$.  Let $a\in
\eftwomm{\ell}, b\in \eftwomm{t}$ be such that 
$\eftwomm{\ell}=\eftwo(a)$, $\eftwomm{t}=\eftwo(b)$. Then there
exists $x\in \eftwom$ such that $x^2+x=a^2b+ab^2$, so that 
$\bs{b}=(b_1,b_2,b_3,b_4):=(1,x,a,b)\in \eftwom^4$ satisfies
$b_1^2b_2+b_1b_2^2=b_3^2b_4+b_3b_4^2$. Furthermore, the entries of
$\bs{b}$ are $\eftwo$-linearly independent.
\end{proposition}

\begin{proof}
To show the existence of $x$, it is
sufficient to prove that $\tr(a^2b+ab^2)=0$ (by Hilbert's Theorem
90). To this end, let $k\in \bbN$ be such that $k\equiv 1\mod(\ell)$
and $k\equiv -1 \mod (t)$ (such a $k$ exists by the Chinese Reminder
Theorem), and note that by definition, $\ell|(k-1)$, and
$t|(k+1)$. Hence 
$$
(ab^2)^{2^k}=a^{2^k}b^{2^{k+1}}=a^2b,
$$
so that $ab^2$ and $a^2b$ are Galois conjugates. This implies
$\tr(ab^2)=\tr(a^2b)$, as required.

For linear independence, note first that $\{1,a,b\}$ are clearly
linearly independent, since
$\eftwomm{\ell}\cap\eftwomm{t}=\eftwo$. Hence, if $\{1,x,a,b\}$ are
linearly dependent over $\eftwo$, then $x=\veps_1+\veps_2 a+\veps_3
b$ for some $\veps_i\in \eftwo$, $i\in\{1,2,3\}$. Substituting in
$x^2+x=a^2b+ab^2$, we 
obtain 
$$
\veps_2(a^2+a)+\veps_3(b^2+b)+a^2b+ab^2=0.
$$
Hence, for any choice for $\veps_2,\veps_3$, we obtain that $a$ is a
root of a polynomial of degree $2$ over $\eftwo(b)$ (specifically,
$X^2(b+\veps_2)+X(b^2+\veps_2)+\veps_3(b^2+b)$) and also $b$ is a root
of a polynomial of degree $2$ over $\eftwo(a)$. This implies that both
$[\eftwo(a,b):\eftwo(a)]\leq 2$ and $[\eftwo(a,b):\eftwo(b)]\leq
2$. Considering the following diagram of field extensions,
this contradicts $\max\{\ell,t\}\geq 3$, which follows from the
assumptions. 


$$
\begin{gathered}
\xymatrix{
& \eftwo(a,b)=\eftwom \ar@{-}[ld]_{t}
\ar@{-}[rd]^{\ell}\\
\txt{$\eftwo(a)$\\$=\eftwomm{\ell}$}\ar@{-}[dr]_{\ell} && 
\txt{$\eftwo(b)$\\$=\eftwomm{t}$}\ar@{-}[dl]^{t}\\ 
 & \eftwo(a)\cap\eftwo(b)=\eftwo
}
\end{gathered}
$$

\end{proof}

\subsubsection{The case $i=3$, $m$ even}\label{sec:ci3}
Recall that for $i=3$, Equations (\ref{eq:main}) read
\begin{eqnarray}
b_1^2b_2+b_1b_2^2 + b_3^2b_4+b_3b_4^2 + b_5^2b_6+b_5b_6^2 & = & 0,
\label{eq:eq31}\\ 
b_1^4b_2+b_1b_2^4 + b_3^4b_4+b_3b_4^4 + b_5^4b_6+b_5b_6^4 & = & 0.
\label{eq:eq32} 
\end{eqnarray}

For this case, a probabilistic algorithm for even $m$ is presented in
Proposition \ref{prop:i3evenm} below. 
We will first need the following lemma, similar to Proposition
\ref{prop:i2}. 

\begin{lemma}\label{lemma:x1}
Suppose that $m\geq 6$ is even. Let $\bs{b}:=(1,y,c,c^2y)$, where
$y\notin \ef_4$ and $c$ is a primitive element of $\ef_4$. Then the
coordinates of $\bs{b}$ are $\eftwo$-linearly independent, and 
\begin{equation}
\label{eq:oddj3}
f_j(b_1,b_2)+f_j(b_3,b_4)=
\begin{cases}
0 & j \text { even }\\
c^2y + cy^{2^j} & j \text{ odd.}
\end{cases}
\end{equation}
\end{lemma}

\begin{proof}
We will prove more generally that for all $x,y\in \eftwom$,
$\bs{b}:=(x,y,cx,c^2y)$ satisfies
\begin{multline}
    \label{eq:oddj2}
f_j(b_1,b_2)+f_j(b_3,b_4)=\\
\begin{cases}
0 & j \text { even }\\
(1+c)x^{2^j}y + (1+c^2)xy^{2^j} & j \text{ odd,}
\end{cases}
\end{multline}
and that if $x,y$ are linearly independent over $\ef_4$, then $\bs{b}$
has $\eftwo$-independent coordinates. The lemma will then follow since
for a primitive $c\in \ef_4$, it holds that $c^2=c+1$.

Suppose that $j$ is even. Then 
\begin{multline*}
f_j(b_3,b_4)=(cx)^{2^j}c^2y+ cx(c^2y)^{2^j} =\\
c^{2^j+2}x^{2^j}y+c^{2^{j+1}+1}xy^{2^j} =
x^{2^j}y+xy^{2^j}=f_j(b_1,b_2),  
\end{multline*}
where the third equality follows since for even $j$, 
$2^j+2,2^{j+1}+1\equiv 0\mod (3)$.

Next, suppose that $j$ is odd. Then 
$$
f_j(b_3,b_4)=(cx)^{2^j}c^2y+ cx(c^2y)^{2^j} = cx^{2^j}y + c^2xy^{2^j},
$$ 
where the last equality follows since for odd $j$, 
$2^j+2\equiv 1\mod (3)$ and $2^{j+1}+1\equiv 2\mod (3)$.

Finally, assume that $x,y$ are linearly independent over $\ef_4$, and
that $a_1x+a_2y+a_3cx+a_4c^2y=0$ for some
$a_1,\ldots,a_4\in \eftwo$. Then $x(a_1+a_3c)+y(a_2+a_4c^2)=0$, which,
as in Proposition \ref{prop:i2}, implies that $a_1=\cdots=a_4=0$.  
\end{proof}

We will also need the following slight modification of \cite[Lemma
11]{K19}. In \cite{K19}, $S$ of Lemma \ref{lemma:kopparty}
below can be either a subspace or an arithmetic progression. However,
if only the subspace case is considered, the statement can be
slightly sharpened.  We omit the proof. 

\begin{lemma}{\cite{K19}}\label{lemma:kopparty} 
Let $q$ be a power of a prime $p$, let $S\subseteq \efq$ be an
$\efpi$-subspace (equivalently, an additive subgroup), and let
$H\subseteq \efq^*$ be a multiplicative subgroup. Then $|H\cap
S|/|S|\geq |H|/q-\sqrt{q}/|S|$.
\end{lemma}

With Lemma \ref{lemma:x1} and Lemma \ref{lemma:kopparty} , we are in a
position to find a probabilistic solution for (\ref{eq:eq31}),
(\ref{eq:eq32}).  

\begin{proposition}\label{prop:i3evenm}
Suppose that $m\geq 6$ is even, and let $c$ be a primitive element of
$\ef_4\subseteq \eftwom$. 
\begin{enumerate}

\item When $y$ is drawn uniformly
from $\eftwom$, the probability that
$c^2y+cy^2$ has a non-zero third root is at least 
$$
p:=\frac{1}{3}-\frac{2}{\sqrt{q}}-\frac{1}{3q},
$$   
where $q:=2^m$.

\item For all $y\notin \ef_4$ for which a third root exists, say,
$c^2y+cy^2=1/d^3$ for some $d$, 
$\bs{b}:=(1,c,d,dy,dc,dc^2y)$ solves (\ref{eq:eq31}) and
(\ref{eq:eq32}) and has independent coordinates.
Hence, with probability at least $p-3/2^m=1/3-O(2^{-m/2})$, a
uniformly drawn $y\in\eftwom$ results in a solution to (\ref{eq:eq31})
and (\ref{eq:eq32}) with independent coordinates. 
\end{enumerate}
\end{proposition}

\begin{proof}
1. The polynomial $f(Y):=c^2Y+cY^2$ is $\eftwo$-linearized with a
one-dimensional kernel, $\{0,c\}$. Hence $S:=f(\eftwom)$ is an
$\eftwo$-subspace of dimension $m-1$. Also, $H:=\{x^3|x\in\eftwom^*\}$
is a multiplicative subgroup of order $|H|=(2^m-1)/3$ (recall that $m$
is even). From Lemma \ref{lemma:kopparty}, it follows that $|H\cap
S|/|S|\geq |H|/q-\sqrt{q}/|S|$. As $f$ defines a linear map, $f(y)$ is
distributed uniformly over $S$ when $y$ is distributed uniformly over
$\eftwom$, and the first assertion follows.

2. Substituting $b_1=1,b_2=c$ in (\ref{eq:eq31}) and (\ref{eq:eq32}),
we obtain
\begin{eqnarray}
b_3^2b_4+b_3b_4^2 + b_5^2b_5+b_5b_6^2 & = & 1
\label{eq:eq31t}\\ 
b_3^4b_4+b_3b_4^4 + b_5^4b_5+b_5b_6^4 & = & 0
\label{eq:eq32t} 
\end{eqnarray}
From Lemma \ref{lemma:x1}, we know that
$\bs{b}':=(b_3,b_4,b_5,b_6)=(1,y,c,c^2y)$ solves (\ref{eq:eq32t}), and
that for $\bs{b}'$, the left-hand side of (\ref{eq:eq31t}) equals
$c^2y+cy^2=1/d^3$. Hence $d\bs{b'}$ satisfies both (\ref{eq:eq31t}) and
(\ref{eq:eq32t}).  

Let $Y$ be the set of all $y$ with a non-zero cube root for
$c^2y+cy^2$. For $y\in Y\smallsetminus \ef_4$, suppose
that there exist $a_1,\ldots,a_6\in \eftwo$ such 
that $a_1+a_2c+a_3d+a_4dy+a_5dc+a_6dc^2y =0$. Then
$$
dy(a_4+a_6c^2)+d(a_3+a_5c)+a_1+a_2c=0.
$$
It is therefore sufficient to show that $1,d,dy$ (where
$1/d^3=c^2y+cy^2$) are linearly independent over $\ef_4$. First, since
we exclude all $y\in \ef_4$, we assure that $d\notin \ef_4$, for if
$d\in \ef_4$, then  $c^2y+cy^2=1$, and $y\in\{1,c^2\}\subset \ef_4$. 

Suppose that there exist $a_1,a_2,a_3\in \ef_4$ such that
$a_1+a_2d+a_3dy=0$. It is easily verified that if one of the $a_i$ is
zero, then so are the other two. For example, if $a_2=0$, then
$a_1=a_3dy$. If either $a_1=0$ or $a_3=0$, then 
we are done. Otherwise, if both $a_1$ and $a_3$ are non-zero, then
raising both sides to the power of $3$, we obtain $y^3=1/d^3$, that
is, $y^3=c^2y+cy^2$, which is equivalent to $y(c^2+cy+y^2)=0$. The
solutions to this equation are $0,1,c^2$, which all lie in $\ef_4$,
and we obtain a contradiction for $y\notin \ef_4$. 

Suppose therefore that $a_1,a_2,a_3$ are all non-zero. Raising both
sides of $1/d=(a_2+a_3y)/a_1$ to the power of $3$, we obtain
$$
c^2y+cy^2=1+y^3+a_2^2a_3y+a_2a_3^2y^2,
$$
that is,
$$
y^3+y^2(c+a_2a_3^2)+y(c^2+a_2^2a_3)+1=0.
$$
This is an equation of the form $y^3+ay^2+a^2y+1=0$ for some
$a\in\ef_4$. If $a=0$, then the equation reads $y^3=1$, and there are
no solutions for $y\notin \ef_4$. Otherwise, if $a\neq 0$, then
the equation is $(y+a)^3=0$, which again leads to a contradiction
for $y\notin \ef_4$.
\end{proof}

\subsubsection{The case $i=4$, with $4|m$}\label{sec:ci4}
For $i=4$, the set of equations (\ref{eq:main}) has the following form:
\begin{eqnarray}
f_1(b_1,b_2)+f_1(b_3,b_4)=f_1(b_5,b_6)+f_1(b_7,b_8)
\label{eq:eq41}\\ 
f_2(b_1,b_2)+f_2(b_3,b_4)=f_2(b_5,b_6)+f_2(b_7,b_8)
\label{eq:eq42}\\
f_3(b_1,b_2)+f_3(b_3,b_4)=f_3(b_5,b_6)+f_3(b_7,b_8)
\label{eq:eq43} 
\end{eqnarray}

\begin{proposition}\label{prop:i4}
Suppose that $m\geq 8$ and $4|m$, so that $\ef_{16}\subset \eftwom$. Let
$y\in \eftwom\smallsetminus \ef_{16}$, and let $c\in \ef_4$ be
primitive. Let $d\in\ef_{16}$ be a primitive $5$-th root of unity. Then 
$$
\bs{b}=(b_1,\ldots,b_8):=(1,y,c,cy,d,dy,dc,dcy)
$$
solves (\ref{eq:eq41})--(\ref{eq:eq43}) and has $\eftwo$-linearly
independent coordinates.
\end{proposition}

\begin{proof}
Note that 
\begin{equation}\label{eq:homog}
f_j(b_5,b_6)+f_j(b_7,b_8) = d^{2^j+1} \big(f_j(b_1,b_2)+f_j(b_3,b_4) \big)
\end{equation}
for all $j$. By Proposition \ref{prop:i2}, we have
$f_j(b_1,b_2)+f_j(b_3,b_4)=0$ for $j=1,3$.
This means that for $\bs{b}$ defined in the proposition, both sides of
(\ref{eq:eq41}) and (\ref{eq:eq43}) are zero, and since $d^5=1$, it
follows from (\ref{eq:homog}) that (\ref{eq:eq42}) also holds.

To prove independence, suppose that for $a_1,\ldots,a_8\in \eftwo$, it
holds that 
$$
a_1+a_2y+a_3c+a_4cy+a_5d+a_6dy+a_7dc+a_8dcy=0.
$$
Then
$$
y\big(d(a_6+a_8c)+a_2+a_4c)\big)+\big(d(a_5+a_7c)+a_1+a_3c\big)=0.
$$
This is an equation of the form $\gamma y +\delta=0$ for
$\gamma,\delta\in \ef_{16}$. Since $y\notin\ef_{16}$, we must have
$\gamma=\delta=0$. 
As $1,d$ are $\ef_4$-linearly independent ($d\notin \ef_4$), we
conclude that $a_1=\cdots =a_8=0$, as required.
\end{proof}

\subsection{Explicit support of a minimum weight
codeword}\label{sec:usingdct} 
Write $\bs{c}(i,0)$ for the evaluation vector of the polynomial $f$ in
Theorem \ref{thm:basis}, which is a minimum-weight
codeword of $\ebbch{2^{m-1}-2^{m-1-i}}$. For $s\in\{1,\ldots, m-2i\}$,
let $\bs{c}(i,s)$ be the minimum-weight codeword of
$\ebbch{2^{m-1-s}-2^{m-1-i-s}}$ constructed from $\bs{c}(i,0)$ through the
down-conversion theorem  (Theorem \ref{thm:dconversion}), as indicated
in Remark \ref{rem:alls}. Next, we would like to specify in detail the
support of $\bs{c}(i,s)$ for all $s$.

Using the notation of Theorem \ref{thm:basis}, write
$V_0:=\linspan_{\eftwo}(b'_{2i+1},\ldots,b'_m)$, and let $S_0\subset
\linspan_{\eftwo}(b'_1, \ldots, b'_{2i})$ be defined as
the support of $\tr(b_1X)\tr(b_2X)+\cdots +\tr(b_{2i-1}X)\tr(b_{2i}X)$
in $\linspan_{\eftwo}(b'_1, \ldots, b'_{2i})$. 
By the definition of $\bs{c}(i,0)$, the support $S$ of $\bs{c}(i,0)$
is given by $S=S_0+V_0$.  
In the general case of $s\geq 0$, we choose some subspace
$V\subseteq V_0$ of dimension $s$ (for $s\leq \dim(V_0)=m-2i$), and
w.l.o.g.~we will assume that
$V=\linspan_{\eftwo}(b'_{2i+1},\ldots,b'_{2i+s})$. 
Recall that $A(X)$ is defined as the annihilator of $V$,
and that $A(X)$ is a linearized polynomial. Note that 
finding the $s$ unknown coefficients defining $A(X)$ amounts to
solving a system of $s$ linear equations in $s$ unknowns, and hence
has complexity $O(s^3)$.\footnote{Writing
$A(X)=X^{q^s}+a_{s-1}X^{q^{s-1}}+\cdots+a_0X$ for the annihilator of
the subspace with basis $\gamma_1,\ldots,\gamma_s\in\eftwom$, it is
readily verified that 
$$
\begin{pmatrix}
\gamma_1 & \gamma_1^q & \cdots & \gamma_1^{q^{s-1}}\\
\gamma_2 & \gamma_2^q & \cdots & \gamma_2^{q^{s-1}}\\
\vdots & \vdots & \ddots & \vdots\\
\gamma_s & \gamma_s^q & \cdots & \gamma_s^{q^{s-1}}\\
\end{pmatrix}
\begin{pmatrix}
a_0\\
a_1\\
\vdots\\
a_{s-1}
\end{pmatrix}
=\begin{pmatrix}
\gamma_1^{q^s}\\
\gamma_2^{q^s}\\
\vdots\\
\gamma_2^{q^s}
\end{pmatrix},
$$
and the {\it Moore matrix} on the left is invertible, since the
$\gamma_j$ are linearly independent \cite{LN97}.}

\begin{remark}\label{rem:indep}
{\rm
Note that for any distinct indices
$j_1,\ldots,j_{\ell}\in\{1,\ldots,m\}\smallsetminus 
\{2i+1,\ldots,2i+s\}$ (for some $\ell\leq m-s$), it holds that 
$A(b'_{j_1}), \ldots, A(b'_{j_{\ell}})$ are $\eftwo$-linearly
independent, since otherwise $A(X)$ would have roots outside $V$,
contradicting its definition.
}   
\end{remark}

By Theorem \ref{thm:dconversion}, the support of $\bs{c}(i,s)$ is 
given by
\begin{multline}\label{eq:as}
\supp(\bs{c}(i,s))=A(S)=\\
A(S_0)+\linspan_{\eftwo}(A(b'_{2i+s+1}),
\ldots, A(b'_m)). 
\end{multline}

To continue, we will use the following notation: 
for a row or column vector $\bs{x}$, we will write $\setv(\bs{x})$ for
the set of all entries of $\bs{x}$. For example,
$\setv\big((1,2,2,3,4,4)\big)=\{1,2,3,4\}$.
Recall that the evaluation vector of $(x_1,\ldots,x_{2i})\mapsto
x_1x_2+\cdots +x_{2i-1}x_{2i}$ on $\eftwo^{2i}$ has weight
$w:=2^{2i-1}-2^{i-1}$. 
Let $M\in \eftwo^{w\times 2i}$ be a matrix whose rows
are all the vectors $(x_1,\ldots, x_{2i})$ for which $x_1x_2+\cdots
+x_{2i-1}x_{2i}=1$.  Then the set $S_0$ defined above is given by
$$
S_0=\setv\left( M\cdot \begin{pmatrix}
b'_1\\
\vdots\\
b'_{2i}
\end{pmatrix}
\right),
$$   
and from (\ref{eq:as}), 
\begin{multline}\label{eq:finalsupp}
\supp(\bs{c}(i,s))=\setv \left(M\cdot \begin{pmatrix}
A(b'_1)\\
\vdots\\
A(b'_{2i})
\end{pmatrix}\right) +\\
 \linspan_{\eftwo}(A(b'_{2i+s+1}),\ldots,
A(b'_m)). 
\end{multline}

\begin{remark}\label{rem:bool}
{\rm
If we complete
$A(b'_1),\ldots,A(b'_{2i}),A(b'_{2i+s+1}),\ldots,A(b'_m)$ to a 
basis $D$ for $\eftwom/\eftwo$,\footnote{Note that the specified elements
are linearly independent by Remark \ref{rem:indep}.} say, by elements
$b''_{2i+1},\ldots,b''_{2i+s}$ corresponding to coefficients
$x_{2i+1},\ldots,x_{2i+s}$, then the support in (\ref{eq:finalsupp}) is
the support of
\begin{multline*}
\Big((x_1,\ldots,x_m)\mapsto\\
(x_1x_2+\cdots+x_{2i-1}x_{2i})(1+x_{2i+1})\cdots(1+x_{2i+s})
\Big)_D^{\mathrm{fld}},
\end{multline*}
similarly to \cite{KaLi72}.
}
\end{remark}

\subsubsection{Examples}\label{sec:examples} 

\begin{example}\label{eg:i3}
{\rm
In this example we give the concrete support of minimum-weight
codewords of {\bf non}-extended BCH codes of designed
distance $27$, for several values of $m$. This case
corresponds to the parameters $i=3$ and $s=m-2i=m-6$. Note that the
support for a codeword of the non-extended code can be obtained as
follows: 1.~Find the support $S$ of a codeword of the extended code.
2.~Take $(x+S)\setminus \{0\}$ for an arbitrary element
$x\in S$. The supports in Table \ref{table:t27}
are described as follows. First, we fix some primitive element
$\alpha\in \eftwom$, which we specify by its minimal polynomial. Then,
an index $j\in \{0,\ldots, 2^m-2\}$ appears in the ``log support''
column iff  $\alpha^j$ is in the support. Note also that the support
for even $m$ was obtained using Section \ref{sec:ci3}, while the
support for odd $m$ was obtained using the heuristic method from
Appendix \ref{app:heuristics}. 

\begin{table*}[h]
\centering
\begin{tabular}{|c|c|c|}
\hline
$m$ & primitive polynomial & log support \\
\hline
8 & $X^8+X^4+X^3+X^2+1$ &
$30,37,42,51,57,61,74,77,90,$\\
&&
$91,99,115,121,134,136,146,154,$\\
&&
$162,176,193,207,227,231,239,242,244,245$
\\
\hline
9 & $X^9+X^4+1$ & $21,52,54,58,108,109,122,160,193,$
\\
&& $195,197,204,218,236,238,247,276,292,$
\\
&& $312,374,381,391,396,411,413,479,503$
\\
\hline
10 &$X^{10}+X^3+1$  & $46, 48, 108,118, 160, 251, 330, 341,346,$ 
\\
&& $366,385,389,451,452,459,541,671,682,$
\\
&& $687,693,726,793,800,842,869,933,944$
\\
\hline
11 & $X^{11}+X^2+1$ & $5,19,29,77,128,188,245,256,291,$
\\ 
&& $356,524,737,832,906,992,1182,1187,1233,$
\\ 
&& $1284,1422,1453,1507,1514,1631,1740,1782,1875$
\\
\hline
12 &  $X^{12}+X^6+X^4+X+1$ & $186,325,477,500,707,749,1036,1077,1147,$
\\ 
&& $1225,1609,1637,1733,1842,2367,2401,2442,2853,$
\\
&& $2916,2974,3002,3103,3207,3230,3437,3734,3877$
\\
\hline
13 &  $X^{13}+X^4+X^3+X+1$ & $1104,1213,1261,1381,1823,2044,3254,4896,4982,$
\\ 
&& $5010,5017,5498,5722,5866,5902,6527,7003,7196,$
\\
&& $7212,7579,7648,7725,7731,7793,7810,7960,8093$
\\
\hline
14 &  $X^{14}+X^{10}+X^6+X+1$ & $75,2064,3354,4180,4411,5074,5536,5559,6082,$
\\ 
&& $7083,7525,7586,8815,9600,9872,10664,11020,11180,$
\\
&& $12496,12544,13047,13956,14276,15061,15798,15996,16245$
\\
\hline
15 &  $X^{15}+X+1$ & $458,508,894,1188,1453,2023,4522,6610,10300,$
\\ 
&& $10946,11107,11370,12145,12635,14841,18805,19244,19780,$
\\
&& $19915,20926,23454,24062,25744,26911,27453,27851,29367$
\\
\hline
16 &  $X^{16}+X^{12}+X^3+X+1$ & $1535,3131,5026,6975,7701,7713,14337,15167,24718,$
\\ 
&& $28238,29546,29558,34509,36182,37012,37492,45225,46563,$
\\
&& $46821,50072,51391,53448,53461,53989,56354,59337,62153$
\\
\hline

\end{tabular}
\caption{The supports of minimum-weight words in $\bbch{27}$ for several
values of $m$. Presented numbers are those values of $j$ for which
$\alpha^j$ is in the support, where $\alpha$ is a primitive element
with the specified minimal polynomial.}
\label{table:t27}
\end{table*}
}
\end{example}

\begin{example}
{\rm
Take $m=16$, $i=2$, $s=10$. After puncturing (as described in Example
\ref{eg:i3}), we obtain a codeword of weight $23$ in
$\bbch{23}\subseteq \eftwo^{2^{16}-1}$ with the support specified in
Table \ref{table:t23}. As
above, the support is specified by those $j$ for which $\alpha^j$ is
in the support, where $\alpha$ is a root of the primitive polynomial
$X^{16}+X^{12}+X^3+X+1$.
\begin{table}[h!]
\centering
\begin{tabular}{|c|}
\hline
$613,7637,10903,13152,13971,14915,14983,$\\
$18473,30809,30977,37218,39604,40125,41649,$\\
$48399,48563,51105,51306,$ $53563,55559,55823,$\\ $56625,$ $62722$ 
\\ \hline
\end{tabular}
\caption{The support of a minimum-weight word in $\bbch{23}$ for
$m=16$. Presented numbers are those values of $j$ for which $\alpha^j$ is 
in the support, where $\alpha$ is a root of $X^{16}+X^{12}+X^3+X+1$.}
\label{table:t23}
\end{table}
}
\end{example}

\section{A construction via Gold functions}\label{sec:gold} 
The Boolean functions  $(x_1,\ldots,x_{2i})\mapsto x_1x_2+\cdots +
x_{2i-1}x_{2i}$ we used to construct minimum-weight words are
examples of quadratic {\it bent functions}
\cite[Sec.~14.5]{MacSl} $\eftwo^{2i}\to \eftwo$.\footnote{In fact, by
Dickson's Theorem \cite[Thm.~15.2.4]{MacSl}, \emph{all} quadratic bent
functions are affine equivalent to this function.}
This raises the natural question whether {\it Gold
functions} (see, e.g., \cite{Le06}), which are quadratic bent
functions $\ef_{2^{2i}}\to \eftwo$ 
that do not rely on a choice of basis, can also be used to construct
minimum-weight words. We answer this question in the
affirmative in Proposition \ref{prop:gold}, by showing that Gold
functions can be used to construct minimum-weight words when
$2i|m$. We note that the results of this section include cases not
covered in Section \ref{sec:main}; see Table \ref{table:cover}.   

While in Section \ref{sec:main} we first constructed words in BCH
codes of large designed distance (with $s=0$) and then used the
down-conversion theorem (Theorem \ref{thm:dconversion})
to cover lower designed distances, in this section we work the other
way around: The construction of Proposition \ref{prop:gold} is for BCH
codes of small designed distance, with $s=m-2i$. The up-conversion
theorem (Theorem \ref{thm:uconversion}) can then be used to cover
higher distances.

Let us first briefly recall the definition of bent functions in
general, and Gold functions in particular. 

\begin{definition}{(Bent functions)}
{\rm
For $i\in \bbNp$, a function $f\colon \eftwo^{2i}\to \eftwo$ is called
{\bf bent} if the Hamming distance between the evaluation vector
$\bs{f}$ of $f$ and the evaluation vector $\bs{g}$ of any affine
functional $g\colon \eftwo^{2i}\to\eftwo$ is at least $2^{2i-1}-2^{i-1}$.
}
\end{definition}
In fact, if $f$ is bent, then it is known that the Hamming distance
between $\bs{f}$ and the evaluation vector $\bs{g}$ of any affine
functional is  $2^{2i-1}\pm 2^{i-1}$ \cite[Theorem 14.5.6]{MacSl}.

\begin{definition}{(Gold functions)}\label{def:gold}
{\rm
For $i,r\in\bbNp$, let $d:=2^r+1$, and suppose that $\gcd(d,2^{2i}-1)>1$, so
that the function $\ef_{2^{2i}}\to\ef_{2^{2i}}, x\mapsto x^d$ is not
surjective. Let $\beta\in\ef_{2^{2i}}$ be an element that does not have a $d$-th
root. Then the function $\ef_{2^{2i}}\to\eftwo, x\mapsto\tr(\beta
x^d)$ is called a {\bf Gold function}.
}
\end{definition}
For a proof that Gold functions are bent, see \cite[Sec.~2A]{Le06}. We
note, however, that this will not be required for the proof of
Proposition \ref{prop:gold}.

It will be useful to introduce some additional notation that is
relevant only for the current section. For positive integers
$a|b$, let us write 
$\trv{b}{a}$ for the trace map $\ef_{2^{b}}\to\ef_{2^{a}}$,
$x\mapsto \sum_{j=0}^{b/a-1} x^{(2^{a})^j}$, and
$\normv{b}{a}$ for the {\it norm map}
$\ef_{2^{b}}\to\ef_{2^{a}}$, $x\mapsto \prod_{j=0}^{b/a-1}
x^{(2^{a})^j}$. Recall that for $a|b|c$, it holds that
$\trv{c}{a}=\trv{b}{a}\circ \trv{c}{b}$
(e.g., \cite[Theorem VI.5.1]{Lang}).

In the following proposition, we will use one particular Gold
function from the family defined in Definition \ref{def:gold}. 

\begin{proposition}\label{prop:gold}
For $i\in \bbNp$, let $\beta\in \ef_{2^{2i}}\smallsetminus
\ef_{2^i}$. Let $f_1(X):=T(\beta X^{2^i+1})\in \ef_{2^{2i}}[X]$, where
$T(X):=X+X^2+\cdots+X^{2^{2i-1}}$, and let
$f(X):=f_1(X)+1$. Let $\bs{f}:=(f(x))_{x\in \ef_{2^{2i}}}$ be the
evaluation vector of $f$.  Then the following holds: 
 
\begin{enumerate}
\item  $\bs{f}$ is a minimum weight codeword of
$\ebch(2^{2i-1}-2^{i-1})\subset \eftwo^{2^{2i}}$. 

\item Suppose that $2i|m$, let $g(X):=f(X)\cdot \prod_{\gamma\in
\eftwom\smallsetminus \ef_{2^{2i}}}(X-\gamma)\in \eftwom[X]$, and put
$\bs{g}:=(g(x))_{x\in\eftwom}$. Then $\bs{g}$ is a minimum-weight
codeword of $\ebch(2^{2i-1}-2^{i-1})\subset \eftwo^{2^m}$.
\end{enumerate}  

\end{proposition} 

\begin{proof} 
1. Let us first prove that\footnote{Note that as $f_1$ is known
to be bent \cite{Le06}, it is clear that the evaluation vector of either
$f_1$ or $f$ must have the specified weight. For completeness, we
include here a simple direct proof, which is also useful for proving
ahead that $\bs{f}$ is in the required BCH code.}
$\wt(\bs{f})=2^{2i-1}-2^{i-1}$. Equivalently, we
have to prove that the number of \emph{zeros} of $f_1(X)$ in
$\ef_{2^{2i}}$ equals $2^{2i-1}-2^{i-1}$.

First, we note that for all $x\in \ef_{2^{2i}}$,
\begin{eqnarray}
f_1(x) & = & \trv{i}{1}\big(\trv{2i}{i}(\beta \normv{2i}{i}(x))\big)
\nonumber\\
 & = & \trv{i}{1}\big(\normv{2i}{i}(x)\cdot \trv{2i}{i}(\beta)\big)\text{
(linearity of $\tr$)} \nonumber\\
 & = & \trv{i}{1}(c\normv{2i}{i}(x)), \label{eq:intermedtr}
\end{eqnarray}
where $c:=\trv{2i}{i}(\beta)\in \ef_{2^i}$ is
non-zero by the assumption $\beta\notin \ef_{2^i}$. Let $M_c\colon
\ef_{2^i}\to \ef_{2^i}$ be the invertible linear map $y\mapsto c\cdot
y$. From (\ref{eq:intermedtr}), we have 
$$
f_1\colon \ef_{2^{2i}}\stackrel{\normv{2i}{i}}{\longrightarrow}
\ef_{2^i} \stackrel{M_c}{\longrightarrow} \ef_{2^i}
\stackrel{\trv{i}{1}}{\longrightarrow} \eftwo.
$$

Now, $(\trv{i}{1})^{-1}(0)$ is an $\eftwo$-subspace of dimension
$i-1$, and so is $V:=M_c^{-1}\big((\trv{i}{1})^{-1}(0)\big)$. Hence
$|V|=2^{i-1}$. Also, $\normv{2i}{i}$ restricts to a
surjective\footnote{As for all $x$,
$\normv{b}{a}(x)=x^{(2^{b}-1)/(2^{a}-1)}$, it is clear that a
primitive element is mapped to a primitive element, and hence
$\normv{b}{a}$ is surjective.}
homomorphism $\ef_{2^{2i}}^*\to \ef_{2^i}^*$, and hence the inverse
image of each non-zero element of $V$ has size $(2^{2i}-1)/(2^i-1)=2^i+1$,
while the inverse image of $0$ is $\{0\}$. Hence

\begin{eqnarray*}
|f_1^{-1}(0)| &= &|(\normv{2i}{i})^{-1}(V)|\\
 & = &(2^i+1)\cdot
(2^{i-1}-1)+1\\
&= & 2^{2i-1}-2^{i}+2^{i-1}-1+1\\
&=& 2^{2i-1}-2^{i-1}, 
\end{eqnarray*}

as required.

Next, we have to prove that $f$ is equivalent modulo $X^{2^{2i}}+X$ to a
polynomial in $P_{d-1}$ for 
$d:=2^{2i-1}-2^{i-1}$. Clearly, it is sufficient to prove that $f_1$
is equivalent to a polynomial of degree $\leq
2^{2i}-d=2^{2i-1}+2^{i-1}$ modulo $X^{2^{2i}}+X$. Now, it follows from
(\ref{eq:intermedtr}) that $f_1$ is equivalent to
\begin{equation}\label{eq:eich}
h(X):=\sum_{j=0}^{i-1} \big(c\cdot X^{1+2^i} \big)^{2^j},
\end{equation}
and it is clear that $\deg(h)=(1+2^i)2^{i-1}=2^{2i-1}+2^{i-1}$. This
completes the proof of part 1.

2. First note that
\begin{eqnarray*}
\prod_{\gamma\in \eftwom\smallsetminus \ef_{2^{2i}}}(X-\gamma) & = &
\frac{X^{2^m-1}+1}{X^{2^{2i}-1}+1}\\
& = &\sum_{j=0}^{(2^m-1)/(2^{2i}-1)-1}
(X^{2^{2i}-1})^j 
\end{eqnarray*}
takes the constant value $1$ on
$\ef_{2^{2i}}$.\footnote{Alternatively, this could be proved from
Lemma \ref{lemma:subspace}, by substituting $0$ to
verify that the single evaluation result is indeed $1$.}
It is therefore clear from
definition that the evaluation vector of $g$ takes 
value in $\eftwo$, and that $\wt(\bs{g})=\wt(\bs{f})$. Also, the
evaluation vector of $g$ is equal to that of  $\tilde{g}:=(1+h(X))\cdot
\prod_{\gamma\in \eftwom\smallsetminus \ef_{2^{2i}}}(X-\gamma)$ (with
$1+h(X)$ from (\ref{eq:eich}) instead of $f(X)$), and 
\begin{multline*}
\deg(\tilde{g}) = \deg(h)+2^m-2^{2i} = \\ 
2^{2i-1}+2^{i-1}+2^m-2^{2i} = 2^m-(2^{2i-1}-2^{i-1}),
\end{multline*}
as required.
\end{proof} 

Note that in part 2 of the proposition, the multiplication by
$\prod_{\gamma\in \eftwom\smallsetminus \ef_{2^{2i}}}(X-\gamma)$ has
the effect of 
completing the evaluation vector on $\ef_{2^{2i}}$ to an evaluation
vector on $\eftwom$ by padding with zeroes. This is just a special case
of the following observation:
If a code $D$ is obtained by shortening a code $C$, and if $D$ has a
codeword of weight $d$, then so does $C$ -- just pad the codeword of
$D$ with zeros. 

Note also that part 1 actually proves that the evaluation vectors of
$f$ and $f_1$ are bent: evaluation vectors of affine functionals such
as $x\mapsto \tr(\gamma x)+\veps$ (for $\gamma\in \ef_{2^{2i}}$ and
$\veps\in\eftwo$) are clearly in the BCH code (as they can be
presented as evaluation vectors of polynomials of a low enough
degree), and hence the distance of $\bs{f}$ from the evaluation of any
affine functional is at least\footnote{Note that $f$ itself is not
affine, as the weight of its evaluation vector is not in
$\{0,2^{2i-1},2^{2i}\}$.} $2^{2i-1}-2^{i-1}$.  From \cite[Theorem
14.6]{MacSl}, it follows that $f$ is bent.  

As {\it any} codeword of
$\ebch(2^{2i-1}-2^{i-1})\subset \eftwo^{2^{2i}}$ is either the evaluation
of an affine function, or 
nonlinear with distance at least $2^{2i-1}-2^{i-1}$ from any affine
functional, it also follows from \cite[Theorem 14.6]{MacSl} that all
the weights in this code are $\{0,2^{2i-1}-2^{i-1},
2^{2i-1},2^{2i-1}+2^{i-1},2^{2i}\}$.\footnote{In fact, it can be
verified that $\ebbch{2^{2i-1}-2^{i-1}}\subset \eftwo^{2^{2i}}$
coincides with the code from Problem 62 on pp.~252--253 of
\cite{MacSl}, specified by its non-zeros.}

\begin{example}
{\rm
Consider the case $i=5$, $m=10$, which is not covered by the results
of Section \ref{sec:main}. Part 1 of Proposition 5.3 states that a
minimum-weight 
codeword in the extended BCH code of length $2^m=2^{10}$ and designed
distance $d(m,0,5)=496$ can by obtained by simply evaluating the Gold
function from the Proposition. 
Consider now the case $i=5$, $m=20$, $s=10$. Part 2 of the same proposition
shows that the support of a  minimum-weight codeword of 
the extended BCH code of length $2^m=2^{20}$ and minimum distance $d(m,10,5)$
(which is again $496$) can be taken as identical to that of the
codeword of length $2^{10}$, when $\eftwomm{10}$ is considered as a
subfield of $\eftwomm{20}$ (recall that the support is specified by
elements of $\eftwomm{10}$). Finally, for $s<10$, we may obtain the
support of minimum-weight codewords of length $2^{20}$ and weight
$d(m,s,5)=d(m,10,5)\cdot 2^{10-s}$ by using the up-conversion
theorem. Clearly, the subspace $U$ in the up-conversion theorem can be
taken as any subspace that contains $\eftwomm{10}$.
}
\end{example}

\section{Conclusions and open questions}\label{sec:conclusions}
We have reduced the problem of finding minimum-weight codewords in
$\ebbch{2^{m-1}-2^{m-1-i}}\cap \RM(2,m)$ to the problem of finding a
solution $b_1,\ldots,b_{2i}$ with independent coordinates to the
system (\ref{eq:main}) of $i-1$ equations. By presenting deterministic
or probabilistic solutions for the cases $i=2,3,4$ (with some
limitations on $m$ described above), we have presented $O(m^3)$
algorithms for specifying the support of minimum-weight codewords of
$\ebbch{2^{m-1-s}-2^{m-1-i-s}}$ for the above values of $i$, and for
$s\in\{0,\ldots,m-2i\}$, using the down-conversion theorem. Along
the way, we have re-proved the down- and up-conversion theorems, using
a new method. We have also presented a construction via Gold
functions. Finally, we have shown (Appendix \ref{app:gk})
that the construction of \cite{GK12} is a special case of the current
construction for the case $i=2$. 

There are several open questions for future research:

\begin{itemize}

\item What is the success probability of the algorithm in Appendix
\ref{app:heuristics} for the case $i=3$, $m$ odd? Also, can the
probabilistic algorithm for $i=3$, $m$ even, be replaced by a closed-form
deterministic solution?

\item More generally, it would be interesting to find a solution of
(\ref{eq:main}) with independent coordinates for all values of $m,s,i$
considered in \cite{KaLi72}. In particular, is it possible to use
elimination with \grobner{} bases for finding a solution
\emph{with independent coordinates} in a reasonable complexity? Also, can
closed-form solutions be specified in this way?

\item As shown in Appendix \ref{app:affine}, the constructed codewords
are not affine generators for designed distance $>6$, because they are
evaluation vectors of Boolean functions of a too low degree. Is it
possible to find affine generators by solving equations such as
(\ref{eq:main}) for Boolean functions of higher degree than $2$?

\end{itemize}

\section*{Acknowledgment}
We thank Evgeny Blaichman and Ariel Doubchak for asking us the
question that lead to this research. Their request to
construct minimum-weight codewords for as many binary BCH codes as
possible has motivated and initiated the work on the current results.


\appendices

\section{Proofs of the conversion theorems}\label{app:conv}

\subsection{Proof of the down-conversion theorem}\label{app:proofdc}
Let $K\subseteq L$ be finite fields, 
$V\subseteq L$ be a $K$-vector space, and
$A(X):=\prod_{\gamma\in V}(X-\gamma)$, be the annihilator polynomial
of $V$. Then $A(X)$ defines a $K$-linear map $L\to L$ 
vanishing on $V$, and is therefore constant on cosets of $V$. 

\begin{lemma}\label{lemma:subspace}
Let $K\subseteq L$ be finite fields. Let $V$ be a $K$-sub-vector space
of $L$. Let $f(X)\in L[X]$ be the polynomial 
$$
f(X)=\prod_{\gamma \in L\smallsetminus V} (X-\gamma).
$$
Then $f$ is constant on $V$, i.e., $f(x)=f(y)$ for any $x,y\in V$.
\end{lemma}

\begin{proof}
Let $A(X)$ be the annihilator polynomial of $V$. A degree
argument shows that for all $\gamma\in L$, $A(X)-A(\gamma)$ is the
annihilator polynomial of the coset $V+\gamma$. Clearly,
$A(X)-A(\gamma)$ is constant on cosets of $V$ (as $A(X)$ is). Since
$f(X)$ is the product of the annihilator polynomials of all cosets of
$V$ (except $V$ itself), and each factor is constant on $V$, so is $f$.
\end{proof}

\begin{lemma}\label{lemma:pol}
Let $K\subseteq L$ be finite fields. Let $V$ be a $K$-sub-vector space
of $L$ of dimension $s$, and let $A(X)$ be the annihilator polynomial of
$V$. Suppose that a  polynomial $f(X)\in L[X]$ with 
$\deg(f(X))\leq |L|-1$ takes constant values on cosets of $V$ in
$L$, that is, $f(a+v)=f(a)$ for all $v\in V$ and $a\in L$. Then there
exists $g(X)\in L[X]$ such that $f(X)=g(A(X))$, and in particular,
$\deg(f)=|K|^s\deg(g)$. 
\end{lemma}

\begin{proof}
Let $S_1$ be the set of polynomials of degree $\leq |L|-1$ that are
constant on cosets of $V$, and let $S_2:=\big\{g(A(X))|g\in L[X],
\deg(g)<|L|/|V|\big\}$. Clearly, $S_2\subseteq S_1$. Since the number of
cosets of $V$ in $L$ is $|L|/|V|$, and $|S_1|$ is the number of
functions $\{\text{cosets of }V\}\to L$, we have
$|S_1|=|L|^{\frac{|L|}{|V|}}$. On
the other hand, $|S_2|$ equals the number of distinct choices of $g$
in the definition of $S_2$, which is also $|L|^{\frac{|L|}{|V|}}$
(note that a polynomial of degree up to $|L|/|V|-1$ has $|L|/|V|$ free
coefficients, and if $g_1(A(X))=g_2(A(X))$, then $g_1=g_2$, as
$|\imag(A)|=|L|/|V|>\deg(g_1-g_2)$). Hence $S_1=S_2$, and the proof is
complete. 
\end{proof}

We are now in a position to give a proof of the down-conversion
theorem that works
directly with the evaluated polynomials from $\pdone$.\footnote{Note
that our definition of $\pd$ is different from that of \cite{KaLi72}.}

\begin{proof}[Proof of Theorem \ref{thm:dconversion}]
The existence of the polynomial $g$ follows
from Lemma \ref{lemma:pol}, and $\deg(g)=\deg(f)/2^s$.
Recalling that $A(X)$ defines an $\eftwo$-linear map $\eftwom\to
\eftwom$ with kernel $V$, it is clear that $A(\eftwom)$ is an
$\eftwo$-subspace of $\eftwom$, of dimension $m-s$. Consider first the
polynomial  
$$
\tl{h}(X):=g(X)\cdot \prod_{\gamma\in\eftwom \smallsetminus
A(\eftwom)}(X-\gamma). 
$$
For all $y\in \eftwom\smallsetminus A(\eftwom)$, we obviously have
$\tl{h}(y)=0$. By Lemma \ref{lemma:subspace}, there exists some
$c'\neq 0$ such  that for all $y\in A(\eftwom)$, $\prod_{\gamma\in\eftwom
\smallsetminus A(\eftwom)}(y-\gamma)=c'$. So, for all $x\in \eftwom$, 
$$
\tl{h}(A(x))=c'\cdot g(A(x))=c'\cdot f(x)
$$
is either $c'$ or $0$. Letting $h(X):=\tl{h}(X)/c'$, we see
that $h(y)$ equals $0$ or $1$ for all $y\in A(\eftwom)$, and $0$ for
all $y\notin A(\eftwom)$. Also, since $\deg(g)=\deg(f)/2^s$, 
\begin{multline*}
\deg(h) =  \deg(f)/2^s+(2^m-2^{m-s}) \leq \\
(2^m-d)/2^s+(2^m-2^{m-s})=2^m-\frac{d}{2^s}.
\end{multline*}
Hence,
$$
\deg(h)\leq \Big\lfloor 2^m-\frac{d}{2^s} \Big\rfloor =
n+1-\Big\lceil\frac{d}{2^s} 
\Big\rceil, 
$$
so that $h\in P_{\lceil d/2^s\rceil -1}$.

Next, we have to consider $\wh(\boldh)$. For $y\in A(\eftwom)$, say,
$y=A(x)$ for some $x\in \eftwom$, we have 
\begin{equation}\label{eq:h}
h(y)=1\iff h(A(x))=1\iff f(x)=1.
\end{equation}
There are $\wh(\boldf)$ choices of $x$ with $f(x)=1$, and, by the assumption that 
$f$ is constant on the cosets of $V$, $2^s|\wh(\boldf)$, and the above
choices of $x$ form a union of $\wh(\boldf)/2^s$ cosets of
$V=\ker(A)$. The linear map $A$ takes 
distinct values on distinct cosets of its kernel, and so there are
exactly $\wh(\boldf)/2^s$ distinct values of $y\in A(\eftwom)$ for
which $h(y)=1$ (and, of course, for all $y\notin A(\eftwom)$,
$h(y)=0$).  

Finally, it follows from (\ref{eq:h}) that if $S\subseteq \eftwom$ is the
support of $\boldf$, then the support of $\boldh$ is $A(S)$. 
\end{proof}

\subsection{Proof of the up-conversion theorem}\label{app:proofuc}
While the proof of the up-conversion theorem in \cite{PeWe} is based
on properties of the annihilator of the support of BCH codewords, the
proof below works directly with the evaluated polynomials. 

\begin{proof}[Proof of Theorem \ref{thm:uconversion}]
Let $c\neq 0$ be the constant value that
$\prod_{\gamma\in \eftwom\smallsetminus U}(X-\gamma)$ takes on $U$ (by
Lemma \ref{lemma:subspace}).
By the definition of $g(X)$, the evaluation vector of $f(X)$ equals
that of $h(X):=\frac{1}{c}\cdot g(X)\cdot\prod_{\gamma\in\eftwom\smallsetminus
U}(X-\gamma)$. Also, since by construction $\deg(g)\leq |U|-1$, we have
$\deg(h)\leq 2^m-1$, so that $h(X)=f(X)\in \pdvar{d-1}$. Hence,
$\deg(g)+2^m-|U|=\deg(h)\leq 
n-(d-1)=2^m-d$, from which it follows that $\deg(g)\leq |U|-d$. Hence
$$
\deg(g(B(X)))\leq \frac{2^m}{|U|}(|U|-d)=2^m-d\cdot
\frac{2^m}{|U|}.
$$
As $g(X)$ and $f(X)$ define the same function on $U$, it
holds that for all $x\in \eftwom$, $g(B(x))\in \eftwo$, so that
$g(B(X))\in \pdvar{d\cdot\frac{2^m}{|U|}-1}$.

Finally, if $S\subseteq U$ is the support of $\boldf$, then the
support of the evaluation vector of $g(B(X))$ is $B^{-1}(S)$ (again,
since $g$ and $f$ define the same function on $U$) and
$|B^{-1}(S)|=|S|\cdot |\ker(B)|=\wh(\boldf)\cdot \frac{2^m}{|U|}$
(since $S$ is contained in the image $U$ of $x\mapsto B(x)$).   
\end{proof}

\section{The constructed words are not affine generators for weight
$>6$}\label{app:affine}
In Appendix \ref{app:gk}, we will see that  for designed distance $6$
(i.e., for $i=2$ and $s=m-4$), it follows from \cite{GK12}
that it is possible to 
choose a solution for the system of equations such that the resulting
minimum-weight codeword is an affine generator.  However, the
following proposition shows that for all other cases,  the constructed
minimum-weight codeword is not an affine generator, as it is in an
affine-invariant proper subcode of the extended BCH code in
question. We note that in this section, we use freely some facts on
cyclic codes, as well as the definition of BCH codes by their zeros
as cyclic codes, as appearing, e.g., in \cite[Ch.~7--9]{MacSl}.

\begin{proposition}\label{prop:not}
Maintaining the terminology of Section \ref{sec:usingdct}, suppose that
the condition $(i=2\text{ and } s=m-2i)$ does {\bf not} hold. Then for
fixed $i$ and all sufficiently large $m$,
$\bs{c}(i,s)\in\ebbch{2^{m-1-s}-2^{m-1-i-s}}$ is {\bf not} an affine 
generator.  
\end{proposition}

\begin{proof}
It follows from Remark \ref{rem:bool} that\footnote{This is also noted
in \cite[Coro.~1]{KaLi72}.}  
\begin{multline*}
\bs{c}(i,s)\in C:=\\ 
\RM(2+s,m)\cap \ebch(2^{m-1-s}-2^{m-1-i-s}).
\end{multline*}
As $C$ is the intersection of
affine-invariant codes, it is affine invariant. To show that
$\bs{c}(i,s)$ is not an affine generator, it is sufficient to show
that $C$ is a proper subcode of $\ebch(2^{m-1-s}-2^{m-1-i-s})$. 

Note that the punctured cyclic code obtained from $C$ inherits all
zeros of both intersected punctured cyclic codes.  
For a primitive $\alpha\in \eftwom^*$ whose consecutive
powers define the evaluation points, the set of zeros of the punctured
Reed--Muller code $\RM(2+s,m)^*$ is $\{\alpha^j|1\leq j\leq
2^m-2, \wt_2(j)\leq m-s-3\}$ 
(where $\wt_2(\cdot)$ is the number of $1$'s in the binary expansion;
see, e.g., \cite[Sec.~13.5]{MacSl}), and this is true for any choice
of basis used for identifying $\eftwom$ with $\eftwo^m$. 

For $i=2$ and $s=m-2i=m-4$, $\RM(2+s,m)=\RM(m-2,m)=\ebch(4)$ is just
the extended Hamming code, which includes all extended BCH
codes. Hence $C=\ebch(2^{m-1-s}-2^{m-1-i-s})$, and it is not
automatically ruled out that $\bs{c}(i,s)$ is an affine
generator.\footnote{Indeed, the affine generator of \cite{GK12} is of
form $\bs{c}(2,m-4)$ for some choice of basis, as shown in Appendix
\ref{app:gk}.} 

However, for all other choices of $s$ and $i$, we claim
that $\RM(2+s,m)\cap \ebch(2^{m-1-s}-2^{m-1-i-s})$
is a proper subcode of $\ebch(2^{m-1-s}-2^{m-1-i-s})$. To see this, we
will specify $j$ with $\wt_2(j)\leq m-s-3$ such that (1) $j>
d^*:=2^{m-1-s}-2^{m-1-i-s}-1$ (the designed distance of the punctured
BCH code), and (2) $j$ is minimal in its {\it cyclotomic 
coset} modulo $2^m-1$ (see, e.g., \cite[Ch.~4]{MacSl}).\footnote{The
cyclotomic coset modulo $n=2^m-1$ of an integer $i$ is defined as the
set $\{i\bmod n ,2i\bmod n,\ldots,2^{m-1}i\bmod n \}$. It is a
convenient way to represent the orbit of $\alpha^i$ (for a primitive
$\alpha$) under the action of the Galois group of $\eftwom/\eftwo$.}
This will show that the cyclic punctured code $C^*$ obtained from $C$
has more zeros than $\bch(d^*)$, so that $C^*$ is a
proper subcode. 

Assume first that $s\geq 1$, and that if $i=3$, $s\leq
m-2i-1=m-7$ (excluding the case $s=m-2i$), and if $i=2$, $s\leq
m-2i-2=m-6$ (excluding the case\footnote{The case $s=m-2i$ is excluded
by the proposition's hypothesis for $i=2$.} $s=m-2i-1$). For $s\geq 1$,
the binary expansion 
of $d^*$ (LSB first) is given in Table \ref{table:expa}.
\begin{table*}
$$
\begin{array}{cccccccccccccc}
\mbox{\scriptsize $0$}  & \mbox{\scriptsize $1$} & \mbox{\ldots} &
\mbox{\scriptsize $\begin{matrix}m-4-\\i-s\end{matrix}$}&
\mbox{\scriptsize $\begin{matrix}m-3-\\i-s\end{matrix}$}
&\mbox{\scriptsize $\begin{matrix}m-2-\\i-s\end{matrix}$}& 
\mbox{\scriptsize $\begin{matrix}m-1-\\i-s\end{matrix}$} & \mbox{\scriptsize $m-i-s$} & \ldots &
\mbox{\scriptsize $m-3-s$} & \mbox{\scriptsize $m-2-s$} &
\mbox{\scriptsize $m-1-s$} & \ldots & \mbox{\scriptsize $m-1$}\\ 
(1&1& \ldots & 1 & 1 &1&0 & 1&\ldots&1&1&0&\ldots& 0).
\end{array}
$$ \caption{Binary expansion of $d^*$ in the proof of Proposition
\ref{prop:not}.}
\label{table:expa}
\end{table*} 

Consider the number $j$ with the binary expansion (LSB
first) in Table \ref{table:expb}.
\begin{table*}
$$
\begin{array}{cccccccccccccc}
\mbox{\scriptsize $0$}  & \mbox{\scriptsize $1$} & \mbox{\ldots} &
\mbox{\scriptsize $\begin{matrix}m-4-\\i-s\end{matrix}$}&
\mbox{\scriptsize $\begin{matrix}m-3-\\i-s\end{matrix}$}
&\mbox{\scriptsize $\begin{matrix}m-2-\\i-s\end{matrix}$}& 
\mbox{\scriptsize $\begin{matrix}m-1-\\i-s\end{matrix}$} & \mbox{\scriptsize $m-i-s$} & \ldots &
\mbox{\scriptsize $m-3-s$} & \mbox{\scriptsize $m-2-s$} &
\mbox{\scriptsize $m-1-s$} & \ldots & \mbox{\scriptsize $m-1$}\\ 
(1&1& \ldots & 1 & 0 &0&1 & 1&\ldots&1&1&0&\ldots& 0).
\end{array}
$$ \caption{Binary expansion of $j$ in the proof of Proposition
\ref{prop:not}.}
\label{table:expb}
\end{table*} 
Note that in this binary expansion, it was assumed that
$m-4-i-s\geq 0$, that is, that $s\leq m-4-i$, which follows from our
assumptions (and from $s\leq m-2i$). Evidently,
$\wt_2(j)=m-s-3$, and  $j>d^*$. 
The run of zeros ending at the MSB of $j$ consists of $s+1$
zeros. If $s\geq 2$, this is a run of at least $3$ zeros, and
therefore the unique longest cyclic run of zeros in $j$. As the unique
longest run of zeros is aligned with the MSB, $j$ is minimal in
its cyclotomic coset. If $s=1$, there are two maximal runs of zeros,
both of length $2$, and it can be verified directly that $j$ is
minimal in its cyclotomic coset for $i$ fixed and $m$ large
enough. This completes the proof for the considered case. 

In all other cases,\footnote{In detail, the missing cases are:
1.~$i=3,s=m-2i=m-6$, 2.~$i=2$ and $s=m-2i-1=m-5$, and 3.~$s=0$.} it is
straightforward to specify $j$ that satisfies all conditions for large
enough $m$, and we omit the details. 
\end{proof}

\section{The Grigorescu--Kaufman result in the current
context}\label{app:gk}
In \cite{GK12}, an explicit support for a minimum-weight
codeword of $\ebbch{6}$ is given as
$X:=\{0,1,1+y^4, y+y^2+y^4,y^2+y^3+y^4,y+y^3+y^4\}$, where $y\in
\eftwom$, $m\geq 4$, is ``almost arbitrary'' -- in order to have
$|X|=6$, the specified elements of $X$ have to be distinct, so that,
e.g., $y\in \ef_4$ (when $m$ is even) is not
allowed. To have $|X|=6$, it is clearly sufficient that
$1,\ldots,y^4$ are linearly independent over $\eftwo$, or,
equivalently, that $y$ is not in a subfield of size $\leq 16$
(excluding at most $|\ef_8\cup\ef_{16}|=16+8-2=22$ elements). It was
also proved in \cite{GK12} that when $y$ is drawn uniformly from
$\eftwom$, $X$ supports an affine generator with probability
$1-2^{-O(m)}$. 

Next, we would like to show that for an appropriate choice of basis
and affine transformation, the specified codeword from \cite{GK12}
can be described as the evaluation vector of the Boolean function 
$(x_1,\ldots, x_m)\mapsto (x_1x_2+x_3x_4)(1+x_5)\cdots (1+x_m)$, for
$m\geq 5$. Moreover, we will show how a solution for (\ref{eq:i2}) with
independent coordinates can be extracted from the codeword of \cite{GK12}.

First, if we let    
$$
M:=\begin{pmatrix}
0 & 1& 1 & 0 & 0 & 0\\
0 & 0 & 0 & 1 & 0 & 1\\
0 & 0 & 0 &1 & 1& 0\\
0& 0 & 0 & 0 & 1 & 1\\
0& 0 & 1 & 1 & 1 & 1
\end{pmatrix},
$$
then the vector $\bs{v}:=(1,y,y^2,y^3,y^4)\cdot M$
has the elements of $X$ as its entries. 
Let 
$$
A:=\begin{pmatrix}
0& 0 & 1 & 1 & 0\\
0& 1 & 0 & 1 & 0\\
0& 0 & 0 & 0 & 1\\
1& 0 & 0 & 0 & 1\\
0& 1 & 1 & 1 & 0
\end{pmatrix}
$$
and put
$\bs{b}:=(1,1,0,0,0)^t$.
Let $N\in \eftwo^{5\times 6}$ be a matrix whose columns are the
support of the function $\eftwo^5\to \eftwo$ defined by
$(x_1,\ldots,x_5)\mapsto (x_1x_2+x_3x_4)(1+x_5)$,
say,\footnote{Note that such a matrix is determined up to a column
permutation.} 
$$
N:=\begin{pmatrix}
1 & 1 & 1 & 0 & 1 & 0\\
1 & 1 & 1 & 0 & 0 & 1\\
0& 0 & 1 & 1 & 1 & 1\\
0 & 1 & 0 & 1 & 1& 1\\
0 & 0 & 0 & 0 & 0 & 0
\end{pmatrix} 
$$
A direct calculation shows that
$N=AM+(\bs{b},\cdots,\bs{b})$,\footnote{
In fact, $A$ and $\bs{b}$ were \emph{found} by solving this linear
equation in the entries of $A$ and $\bs{b}$. There is a total of
$25+5=30$ such entries, and also $30$ equations (entries of
$N$). However, it can be verified that the involved linear transformation
of the entries of $A$ and $\bs{b}$ is not invertible, so that the
existence of a solution for the specific $N$ is not trivial, let alone
a solution for which $A$ is invertible.} 
and that $A$ is invertible with inverse
$$
A^{-1}=\begin{pmatrix}
0 & 0 &1 & 1 & 0 \\
1& 0 & 0 & 0 & 1\\
0& 1 & 0 & 0 & 1\\
1& 1 &0 & 0 & 1\\
0& 0 &1 & 0 & 0
\end{pmatrix}
$$
Hence, for $m=5$, the coordinate permutation $\bs{x}\mapsto
A\bs{x}+\bs{b}\in \mathrm{GA}(5,2)$ (where $\bs{x}$ is a column
vector) maps\footnote{Here, $\mathrm{GA}(m,q)$ is the {\it 
 general affine group}, see, e.g., \cite[Sec.~13.9]{MacSl}.} the
vector supported on $X$ to the evaluation vector of
$(x_1x_2+x_3x_4)(1+x_5)$ (see ahead for more details).  

More generally, for $m\geq 5$, we may take 
$$
A_m:=\begin{pmatrix}
A \\
 & I_{m-5}\\
\end{pmatrix} \in\eftwo^{m\times m} 
$$
(where $I_d$ is the $d\times d$ identity matrix and blanks stand for
zeros), and $\bs{b}_m:=(1,1,0,\ldots 0)^t\in\eftwo^m$ to obtain that
the coordinate permutation $\bs{x}\mapsto A_m\bs{x}+\bs{b}_m\in
\mathrm{GA}(m,2)$ maps the vector supported on $X$ to
the evaluation vector of
$(x_1,\cdots,x_m)\mapsto(x_1x_2+x_3x_4)(1+x_5)\cdots (1+x_m)$, as we
shall now explain.  

We note that $A_m^{-1}$ can be
written as
$$
A_m^{-1}=\begin{pmatrix}
A^{-1}\\
& I_{m-5}
\end{pmatrix}.
$$
Let $\bs{y}^+_1\in
\eftwo^{m-5}$ be some extension of $\bs{y}:=(1,y,\ldots,y^4)$ to a
basis for $\eftwom/\eftwo$, and
write $\bs{y}^+:=(\bs{y},\bs{y}^+_1)$.\footnote{With a slight abuse of
notation, we write $(\bs{y},\bs{y}^+_1)$ for the vector obtained by
concatenating the two involved vectors.}
When working with the basis
$\bs{y}^+A_m^{-1}=(\bs{y}A^{-1},\bs{y}^+_1)$, the coordinate
$A_m\bs{x}+\bs{b}_m\in\eftwo^m$ corresponds to the field element
\begin{equation}\label{eq:interp}
\bs{y}^+A_m^{-1}\cdot(A_m\bs{x}+\bs{b}_m) =
\bs{y}^+\bs{x}+\bs{y}^+A_m^{-1}\bs{b}_m. 
\end{equation}
The first summand $\bs{y}^+\bs{x}$ is the field element
corresponding to $\bs{x}$ according to the basis $\bs{y}^+$. By taking
$\bs{x}$'s of the form $(\bs{x}_0^t,0,...,0)^t$ where $\bs{x}_0$ runs
on the columns of $M$, we conclude that if we take the support
specified in \cite{GK12}, apply the translation $z\mapsto
z+\bs{y}^+A_m^{-1}\bs{b}_m$ and then interpret the resulting support
according to the basis $\bs{y}^+A_m^{-1}$, we obtain the support of
$(x_1,\cdots,x_m)\mapsto(x_1x_2+x_3x_4)(1+x_5)\cdots (1+x_m)$. Note
that the first translation is within the automorphism group of
extended BCH codes, and obviously does not change the affine orbit of the
corresponding vector. 

We now turn to finding solutions to 
(\ref{eq:i2}) using \cite{GK12}, and explain how the basis
$\bs{y}^+A_m^{-1}$ for which the word of \cite{GK12} is the evaluation vector of
$(x_1,\cdots,x_m)\mapsto(x_1x_2+x_3x_4)(1+x_5)\cdots (1+x_m)$ can be
used to find such solutions.

Let $\bs{d}=(d_1,\ldots,d_m):=\bs{y}^+A_m^{-1}$,  write
$$
S:=\bs{d}\cdot \begin{pmatrix}N\\ 0_{(m-5)\times 6} \end{pmatrix}
$$
(where $0_{\ell_1\times \ell_2}$ is the $\ell_1\times \ell_2$ all-zero
matrix, for $\ell_1,\ell_2\in \bbNp$), 
and let $\bs{c}$ be the codeword of $\ebch(6)$ supported on the
entries of the vector $S$ (which exists, by (\ref{eq:interp})).
Letting $U:=\linspan_{\eftwo}(S)$, it holds that $\dim(U)=4$, since
$N$ has rank $4$.  Since also $U\subseteq
\linspan_{\eftwo}(d_1,\ldots,d_4)$ by the definition of $S$ and $N$,
it follows that $U=\linspan_{\eftwo}(d_1,\ldots,d_4)$. Using the
terminology of the up-conversion theorem (Theorem
\ref{thm:uconversion}), it following from the definition of $g(X)$ in
the theorem, that $g(x_1d_1+\cdots+x_4d_4)=x_1x_2+x_3x_4$. Also, there
exists some full-rank matrix $V\in\eftwo^{m\times 4}$ such that the
image polynomial $B(X)$ satisfies $B(x_1d_1+\cdots+x_md_m) =
\bs{x}V(d_1,\cdots,d_4)^t$ (where here $\bx$ is the row vector
$(x_1,\ldots,x_m)$).\footnote{The existence of some matrix $V$
is clear from the fact that the image is in
$U=\linspan_{\eftwo}(d_1,\ldots,d_4)$, and $V$ is of full rank since $B$
defines an onto linear function $\eftwom\to U$.}

Writing $x\in\eftwom$ as $x=\bs{x}\bs{d}^t$ for some 
$\bs{x}\in\eftwo^m$ and letting
$(z_1,\ldots,z_4):=\bs{x}V$, we therefore have  
\begin{equation}\label{eq:gbx}
g(B(x)) = g(\bs{x}V(d_1,\cdots,d_4)^t)=z_1z_2+z_3z_4,
\end{equation}
and by the up-conversion theorem, the evaluation vector of $g(B(X))$ 
is a minimum-weight codeword of $\ebbch{2^{m-1}-2^{m-3}}$. 

Let $V_1\in\eftwo^{m\times (m-4)}$ be such that $V^+:=(V|V_1)$ is
invertible (this is possible, as $V$ is of full rank), and write
$\bs{e}^t=(e_1,\ldots,e_m)^t:=(V^+)^{-1}\bs{d}^t$. 
As $(V^+)^{-1}$ is
invertible, $\bs{e}$ is a basis for $\eftwom/\eftwo$. Also, the above
arbitrary element $x\in \eftwom$ can be written as
\begin{multline}\label{eq:ex}
x=\bs{x}\bs{d}^t=\bs{x}V^+\bs{e}^t = \bs{x}(V|V_1)\bs{e}^t =\\
z_1e_1+\cdots+z_4e_4+z_5e_5+\cdots + z_me_m,
\end{multline}
where $(z_5,\ldots,z_m):=\bs{x}V_1$. It now follows from Theorem
\ref{thm:basis}, (\ref{eq:gbx}), and (\ref{eq:ex}), that
if $\{e'_1,\ldots,e'_m\}$ is the dual of $\{e_1,\ldots,e_m\}$, then
$e'_1,\ldots,e'_4$ satisfy (\ref{eq:i2}).

\section{A heuristic algorithm for the case $i=3$, odd $m$}
\label{app:heuristics} 
For the case where $i=3$ and $m$ is odd, we will describe in this
appendix a probabilistic algorithm that performed well in practice for
both odd and even $m$,
although we do not have a lower bound on its probability of
success. To describe this algorithm, we first need the following
proposition. Recall that for $j\in \bbNp$, $f_j(X_1,X_2) =
X_1^{2^j}X_2+X_1X_2^{2^j}$.

\begin{proposition}\label{prop:singlepol}
Let $c_1\in \eftwom^*$ and $c_2\in \eftwom$ (for $m\in \bbNp$).  Then
for $x_1,x_2\in \eftwom^*$ with $x_1\neq x_2$, it holds that 
\begin{multline}\label{eq:equiv}
\Big(f_1(x_1,x_2)=c_1 \text{ and } f_2(x_1,x_2)=c_2\Big)\iff\\
x_1,x_2\in \roots(c_1X^3+c_2X+c_1^2).
\end{multline}
\end{proposition}

\begin{proof}
$\Rightarrow$: Suppose that $x_1^2x_2+x_1x_2^2=c_1$ and
$x_1^4x_2+x_1x_2^4=c_2$. Dividing the first equation by $x_1^3$ and
the second by $x_1^5$, and setting $t:=x_2/x_1$, $u:=c_1/x_1^3$, and
$v:=c_2/x_1^5$, we obtain $t^2+t=u$, and $t^4+t=v$. Since
$(t^2+t)^2+(t^2+t)=t^4+t$, we must have $u^2+u=v$. Substituting $u=c_1/x_1^3$,
$v=c_2/x_1^5$, we see that $c_1x_1^3+c_2x_1+c_1^2=0$. From symmetry,
the same equation holds also with $x_2$ instead of $x_1$.

$\Leftarrow$: Multiplying a polynomial of the form $g(X):=aX^3+bX+c$
(for some $a,c\in \eftwom^*$, $b\in\eftwom$) by $X$,
we obtain a linearized polynomial of degree $4$. This means that the
roots of $g$ are obtained by removing $0$ from an $\eftwo$-vector
space of dimension $\leq 2$. Since we assume that $X^3+(c_2/c_1)X+c_1$
has $2$ distinct roots $x_1,x_2$, its roots must therefore be
$x_1,x_2,x_1+x_2$, and the free coefficient is 
$x_1x_2(x_1+x_2)=c_1$. This shows that $f_1(x_1,x_2)=c_1$. Similarly,
the coefficient of $X$ is
$x_1x_2+x_1(x_1+x_2)+x_2(x_1+x_2)=x_1^2+x_2^2+x_1x_2$, which must
therefore equal $c_2/c_1=c_2/(x_1^2x_2+x_1x_2^2)$. Equating and
re-arranging terms, it follows that $f_2(x_1,x_2)=c_2$.
\end{proof}

Recall that for $i=3$, we are looking for $\eftwo$-independent
elements $b_1,\ldots,b_6\in\eftwom^*$ that satisfy (\ref{eq:eq31}) and
(\ref{eq:eq32}). Proposition \ref{prop:singlepol} suggests the
following probabilistic algorithm:

\begin{enumerate}

\item Draw $4$ independent elements $b_1,\ldots,b_4\in\eftwom^*$; in
detail: 

\begin{itemize}

\item Draw $4$ distinct elements. 

\item Form the $4\times m$ binary matrix $M$ whose rows are the
coefficients of the decomposition of the $4$ drawn elements according
to some basis. 

\item Apply Gaussian elimination to $M$ (in complexity $O(m)$, since
the number of rows is fixed), and
check if the matrix is of full rank. If it is, keep the $4$ drawn
elements. Otherwise, re-draw $4$ elements and repeat the above steps. 

\end{itemize}

\item Calculate $c_1:=b_1^2b_2+b_1b_2^2+b_3^2b_4+b_3b_4^2$,
$c_2:=b_1^4b_2+b_1b_2^4+b_3^4b_4+b_3b_4^4$. If $c_1=0$, return to Step
1.

\item Find all non-zero solutions of $c_1x^4+c_2x^2+c_1^2x=0$. Note
that since the involved polynomial is linearized, this
can be done by linear-algebra methods in complexity $O(m^3)$.

\item If there are less than $3$ non-zero solutions, return to Step 1.

\item If there are $3$ solutions, pick any two distinct solutions as
$b_5,b_6$, and check as above that $b_1,\ldots, b_6$ are linearly
independent (complexity: $O(m)$). If so, output
$b_1,\ldots,b_6$. Otherwise, return to step 1.

\end{enumerate}

\begin{IEEEbiographynophoto}{Amit Berman}
(Senior Member, IEEE) received the Ph.D. degree in
electrical engineering from the Technion---Israel Institute of
Technology, Haifa, Israel, in 2013. He is currently with the Advanced
Flash Solution Laboratory, Samsung Electronics Memory Division, as the
Vice President of Research and Development. He has authored over 30
research articles and holds over 100 U.S./Global issued and pending
patents. He was a recipient of several awards, including the Samsung
Contribution Award, the Hershel Rich Innovation Award, the Mitchell
Grant, the HPI Fellowship, and the International Solid-State Circuits
Conference Recognition. 
\end{IEEEbiographynophoto}

\begin{IEEEbiographynophoto}{Yaron Shany}
received the Ph.D.~degree in Electrical Engineering from Tel Aviv
University in 2004.  He is currently with the Advanced Flash
Solution Lab of Samsung Semiconductor Israel R\&D Center. His research
interests include coding theory and its applications to storage
systems and devices. 
\end{IEEEbiographynophoto}

\begin{IEEEbiographynophoto}{Itzhak Tamo}
received his B.A. degree in Mathematics and his B.Sc. and
Ph.D. degrees in Electrical Engineering from Ben-Gurion University,
Israel, in 2008 and 2012, respectively. From 2012 to 2014, he was a
Post-Doctoral Researcher with the Institute for Systems Research at
the University of Maryland, College Park. Since 2015, he has been with
the Electrical Engineering Department at Tel Aviv University, Israel,
where he was promoted to Associate Professor in 2020. His research
interests include storage systems and devices, coding, information
theory, and combinatorics. He was a recipient of the IEEE
Communication Society Data Storage Technical Committee 2013 Best Paper
Award. He also received the 2015 IEEE Information 
Theory Society Paper Award and the Krill Prize in 2018.
\end{IEEEbiographynophoto}

\end{document}